\documentclass[final,12pt]{elsarticle}

\usepackage{amsmath}
\usepackage{amsthm}
\usepackage{amssymb}
\usepackage[all]{xy}

\journal{Advances in Applied Mathematics}

\begin{document}

\newtheorem{theorem}{Theorem}
\newtheorem{lemma}{Lemma}
\newtheorem{corollary}[theorem]{Corollary}
\newtheorem{definition}{Definition}
\newtheorem{proposition}{Proposition}

\begin{frontmatter}

\title{One-sided asymptotically mean stationary channels}

\author{Fran\c{c}ois Simon \fnref{affiliation}}

\address{Telecom SudParis, 9 rue Charles Fourier, 91011 Evry, FRANCE}
\fntext[affiliation]{Institut Mines-Telecom ; Telecom SudParis ; CITI}
\ead{Francois.Simon@telecom-sudparis.eu}

\begin{abstract}
This paper proposes an analysis of asymptotically mean stationary (AMS) communication channels. A hierarchy based on stability properties (stationarity, quasi-stationarity, recurrence and asymptotically mean stationarity) of channels is identified. Stationary channels are a subclass of quasi-stationary channels which are a subclass of recurrent AMS channels which are a subclass of AMS channels.  These classes are proved to be stable under Markovian composition of channels (e.g., the cascade of AMS channels is an AMS channel). Characterizations of channels of each class are given. Some properties of the quasi-stationary mean of a channel are established. Finally, ergodicity conditions of AMS channels are gathered.
\end{abstract}


\end{frontmatter}



\section{Introduction} \label{SectionIntroduction}

Information theory considers information sources which are random processes and  noisy communication channels which are probability kernels. The communication of a source $X$ with distribution $P_X$ through a noisy channel $\nu$ is described by a joint random process (X,Y) whose distribution $P_{XY}$ is the "hookup" $P_X\nu$ of the source and of the channel. The channel output process has for distribution the marginal $P_Y$. 

Channel coding theorems of Information Theory establish conditions to reliably communicate information through noisy channels, i.e., to reliably estimate $X$ from $Y$. These theorems rely on the Shannon-McMillan-Breiman theorem which holds when the random processes under consideration are asymptotically mean stationary (AMS) and ergodic (see \cite{Girardin05}). Thus the characterization of classes of AMS and ergodic sources and channels has been one of the central questions in the field of information theory.

Analyses of information sources and communication channels often consider two-sided random processes, i.e. random processes are supposed to be associated to invertible shifts. The present paper is devoted to an analysis of one-sided AMS communication channels, i.e., the weaker assumption of non-invertible shifts is made. 

On the base of channel stability properties which are stationarity, quasi-station\-arity, recurrence and asymptotically mean stationarity, a hierarchical classification of AMS channels is proposed: stationary channels are a subclass of  quasi-stationary channels,  quasi-stationary channels form a subclass of recurrent AMS channels  and recurrent AMS channels are a subclass of AMS channels. Each class is proven to be stable under cascading (or Markovian composition).

Characterizations of each channel class are given under the form of necessary and sufficient conditions. It is also proved that if a channel is a family of recurrent AMS (resp. AMS) conditional probabilities then the channel is recurrent AMS (resp. AMS).

The quasi-stationary mean of  AMS channels is defined and it is shown that  a recurrent  AMS (resp. AMS) channel is  dominated (resp. asymptotically dominated ) by its quasi-stationary mean (w.r.t a source). The special cases of  the quasi-stationary mean of a recurrent AMS channel  w.r.t. a stationary source  and of the quasi-stationary mean of an ergodic  recurrent and AMS channel w.r.t. an ergodic and stationary source are studied.

After a brief survey of related works in Section \ref{SectionRelatedWorks}, Section \ref{SectionSourcesAndChannels} reviews the classical formal models of sources and channels. Section \ref{SectionRestrictions} shows that restrictions to tail $\sigma$-fields of sources and channels can be consistently defined. Definitions of stability properties of sources are reviewed in Section \ref{SectionSourcesStabilityProperties}.

Section \ref{SectionChannelsStabilityProperties} reminds definitions and characterizations of stationarity, quasi-stationarity and asymptotically mean stationarity of channels. A definition of channel recurrence is proposed. A necessary and sufficient condition for a channel to be recurrent w.r.t. a source is proved. 

In Section \ref{SectionClassificationOfChannels}, it is proved that the set of AMS channels includes the set of recurrent AMS channels which includes the set of quasi-stationary channels which includes the set of stationary channels.

Sections \ref{SectionRAMSChannels} and \ref{SectionOneSidedAMSChannels} give necessary and sufficient conditions for a channel to be respectively recurrent AMS and AMS. These sections establish that a channel made of a family of recurrent AMS, respectively  AMS, conditional probabilities is recurrent AMS, respectively  AMS.

Cascades of channels are studied in Section \ref{SectionCascadesOfChannels}. The quasi-stationary mean of an recurrent AMS channel with respect to a stationary source is characterized in Section \ref{SectionQuasiStationaryMean}. Section \ref{SectionQuasiStationaryMeanErgodicRAMSChannel} briefly analyses ergodicity of recurrent AMS and AMS channels. Properties of the quasi-stationary mean of an ergodic recurrent AMS channel with respect to an ergodic recurrent AMS source are given in Section \ref{SectionQuasiStationaryMeanErgodicRAMSChannel}.
\section{Related works} \label{SectionRelatedWorks}

The analyses of information sources and of communication channels often consider two-sided random processes, i.e., random processes  are built with invertible shifts.  Numerous results have been established under this assumption. 

\cite{FontanaGrayKieffer81} presents in depth analysis of two-sided AMS channels. Key results are given in that paper: the fact that it is enough to check the AMS property on stationary sources, some characterizations of two sided AMS channels and the justified definition of the stationary mean of an AMS channel. It is proved that a two-sided AMS channel is ergodic if and only if its stationary mean is ergodic.The cascade of two-sided AMS channels is proved to be AMS. The key results are not proved for one-sided channels. Most of those results are based on the fact that, for two-sided processes, the stationary mean of a probability dominates and not only asymptotically dominates the AMS probability. 

An important special case of AMS channels has been identified and analyzed in \cite{KiefferRahe81}: Markov Channels. Two-sided and one-sided Markov channels are proved to be AMS. \cite{KiefferRahe81} proves that indecomposable Markov channels are ergodic. Ergodicity of Markov channels is analyzed in \cite{GrayDunhamGobbi87}, the obtained results cover one-sided and two-sided channels. \cite{Sujan81} investigates the properties of the information quantile capacity, of the Shannon capacity and of the operational capacity of two-sided AMS channels. \cite{Kakihara91}, \cite{Kakihara03} and \cite{Kakihara04} give results on ergodicity of AMS channels assuming an invertible shift, thus targeting  two-sided channel, see also \cite{Kakihara99}. Coding theorems and lemmas given in \cite{Gray11} cover one-sided channels. A sufficient condition of ergodicity for an AMS channel is given by Lemma 2.3 of \cite{Gray11}: if, for any AMS source, the input/output process is "weakly mixing" on products of rectangles, the channel is ergodic.

\section{Sources and Channels}\label{SectionSourcesAndChannels}

This section reviews the definitions of information sources and of noisy communication channels.  Information sources can be described following two equivalent models: random processes and dynamical systems (see \cite{Gray11}). Both models will be used below. Notations follow partly \cite{Gray11} and partly \cite{Kakihara99}.

\begin{definition}
Let $\mathcal{I}$ be a countable index set. Let $(A,\mathcal{B}_A)$ be a measurable space, called the alphabet. Let $(A^\mathcal{I}, \mathcal{B}_{A^\mathcal{I}})$ be the measurable sequence space where  $\mathcal{B}_{A^\mathcal{I}}$ is the $\sigma$-field generated by the sets of rectangles $\{x=(x_i)_{i\in \mathcal{I}} \in A^\mathcal{I} / x_i\in B_i,  B_i\in \mathcal{B}_A \text{ for any }i\in \mathcal{J}\}$, $\mathcal{J}$  finite subsets of $\mathcal{I}$. A source $[A,X]$ is a random process $X= \{ X_i ; i\in \mathcal{I} \}$ with values in  $(A^\mathcal{I}, \mathcal{B}_{A^\mathcal{I}})$. The distribution of the source $[A,X]$ is denoted by $P_X$.
\end{definition}
If the index set $\mathcal{I}$ is $\mathbb{N}$ then the process $X$ is said to be one-sided. If the index set $\mathcal{I}$ is $\mathbb{Z}$ then the process $X$ is said to be two-sided. Let $T_A: A^\mathcal{I} \rightarrow A^\mathcal{I}$ be the shift transform on $A^\mathcal{I}$. For a one-sided process:
$$
T_A(x_0 x_1 \ldots x_i \ldots)=x_1 x_2 \ldots x_{i+1} \ldots
$$
For a two-sided process:
$$
T_A(\ldots x_{-j}\ldots x_0 x_1 \ldots x_i \ldots)=\ldots x_{-j+1}\ldots x_1 x_2 \ldots x_{i+1} \ldots
$$
In the latter case, the shift $T_A$ is invertible. 

$T_A$ is assumed to be $(\mathcal{B}_{A^\mathcal{I}},\mathcal{B}_{A^\mathcal{I}})$-measurable. If $\mu$ is a probability on $(A^\mathcal{I},\mathcal{B}_{A^\mathcal{I}})$ then $(A^\mathcal{I},\mathcal{B}_{A^\mathcal{I}}, T_A,\mu)$ is a dynamical system.
\begin{definition}
A source on the alphabet $(A, \mathcal{B}_A)$ is a dynamical system \\ $(A^\mathcal{I},\mathcal{B}_{A^\mathcal{I}}, T_A,\mu)$.
\end{definition}
Let $\Pi_0$ denote the ``zero-time sampling function'': $\Pi_0(x)=x_0$ for any $x=x_0x_1\ldots x_i\ldots$ (or $x=\ldots x_{-j}\ldots x_0 x_1 \ldots x_i \ldots$ in the two-sided case). The two definitions of a source are equivalent. A dynamical system $(A^\mathcal{I},\mathcal{B}_{A^\mathcal{I}}, T_A,\mu)$ determines a random process $X= \{ X_i ; i\in \mathcal{I} \}$ where $X_i(\omega)=\Pi_0(T_A^i(x))=x_i$ and the distribution of $X$ is $\mu$: $P_X=\mu$.

In the sequel, each model will be used when the most relevant. Moreover, the word source will be used to name either the random process, the dynamical system, the distribution $P_X$ or the probability $\mu$.

Let $(A,\mathcal{B}_A)$ and $(B,\mathcal{B}_B)$  be alphabets,  $(A^\mathcal{I}, \mathcal{B}_{A^\mathcal{I}})$ and $(B^\mathcal{I}, \mathcal{B}_{B^\mathcal{I}})$ be the two corresponding sequence spaces. The shifts (assumed non-invertible) on $(A^\mathcal{I}, \mathcal{B}_{A^\mathcal{I}})$ and $(B^\mathcal{I}, \mathcal{B}_{B^\mathcal{I}})$ are respectively denoted $T_A$ and $T_B$.  

Let $\mathcal{B}_{A^\mathcal{I}\times B^\mathcal{I}}$ be the $\sigma$-field generated by the rectangles $\{ F\times G / F \in A^\mathcal{I} , G \in B^\mathcal{I}\}$. $T_A$ and $T_B$ define a measurable shift $T_{AB}$ on the space $(A^\mathcal{I} \times B^\mathcal{I}, \mathcal{B}_{A^\mathcal{I}\times B^\mathcal{I}})$ where $T_{AB}(x,y)=(T_Ax, T_By)$.

\begin{definition}
A noisy communication channel $[A,\nu,B]$  is a function $\nu: A^\mathcal{I} \times \mathcal{B}_{B^\mathcal{I}} \to [0,1]$ such that:
\begin{itemize}
	\item for any $x\in A^\mathcal{I}$, the set function $G \mapsto \nu(x,G)$ is a probability on the space $(B^\mathcal{I}, \mathcal{B}_{B^\mathcal{I}})$
	\item for any $G\in \mathcal{B}_{B^\mathcal{I}}$, the function $x \mapsto \nu(x,G)$ is measurable.
\end{itemize}
If the noisy channel  $[A,\nu,B]$ takes a source $[A,X]$ as input, it produces as an output an information source  $[B,Y]$.  Let $P_{XY}$ be the distribution of the joint process $(X,Y)$ induced by the source $X$ fed into the noisy channel. $P_{XY}$ is defined on rectangles by: 
$$
\forall F\in \mathcal{B}_{A^\mathcal{I}}, \forall G\in \mathcal{B}_{B^\mathcal{I}}, P_{XY}(F\times G) =\int_F \nu(x,G) dP_X
$$
$P_{XY}$ will also be denoted by $\mu\nu$ (the hookup of $\mu$ and $\nu$ where $\mu=P_X$). 
\end{definition}
In fact a noisy communication channel is a probability kernel. In the sequel, the alphabets $(A,\mathcal{B}_A)$ and $(B,\mathcal{B}_B)$  will  be assumed standard. Then the sequence spaces $(A^\mathcal{I}, \mathcal{B}_{A^\mathcal{I}})$ and $(B^\mathcal{I}, \mathcal{B}_{B^\mathcal{I}})$ are also standard.This assumption is made to ensure that  conditional probabilities defined  are regular (see \cite{Gray09} or \cite{Faden85}). Thus, on such spaces, given a joint random process $(X,Y)$, it will always be possible to define a channel $\nu$, unique $P_X$-a.s., taking $X$ as an input and inducing the joint process $(X,Y)$ (\cite{Gray11}). 

\section{Restrictions of Sources and Channels}\label{SectionRestrictions}

Below, some results on channels will be established thanks to properties holding for process distributions on their tail $\sigma$-fields. The following straightforward lemmas show that it is possible to consistently define the restriction of a source and the restriction of a channel to the corresponding tail $\sigma$-fields.

\begin{lemma}\label{LemmaSourceRestriction}
Let $(\Omega, \mathcal{F}, T , \eta)$ be a dynamical system. Let $\mathcal{F}_\infty$ denote the tail $\sigma$-field $\cap_{n\geq 0}T^{-n}\mathcal{F}$ and $\eta_\infty$ the restriction of $\eta$ to $\mathcal{F}_\infty$. Then $(\Omega, \mathcal{F}_\infty,T ,\eta_\infty)$ is a dynamical system. 
\end{lemma}
\begin{proof}
Since $T^{-1}\mathcal{F}_\infty \subset \mathcal{F}_\infty$, $T$ is $(\mathcal{F}_\infty,\mathcal{F}_\infty)$-measurable.
\end{proof}
\begin{definition}\label{DefinitionSourceRestriction}
The restriction of a source $[A,X]$ to its tail $\sigma$-field is the dynamical system $(A^\mathcal{I}, (\mathcal{B}_{A^\mathcal{I}})_\infty,T_A, (P_X)_\infty)$
\end{definition}
\begin{lemma}\label{LemmaChannelRestriction}
Let $[A,X]$ be a source with distribution $\mu$ and $[A,\nu,B]$ a channel. There exists a probability kernel $\nu': ((A^\mathcal{I},\mathcal{B}_{A^\mathcal{I}}) \times (\mathcal{B}_{B^\mathcal{I}})_\infty) \to [0,1]$ such that:
$$
\nu'(x , . ) = \nu_\infty(x,.) \text{  }\mu_\infty\text{-a.e.}
$$
where $\nu_\infty(x,.)$ is the restriction of the probability $\nu(x,.)$ to the tail $\sigma$-field $(\mathcal{B}_{B^\mathcal{I}})_\infty$.
$\nu'$ is the restriction of the channel $[A,\nu,B]$ and is denoted by $[A,\nu_\infty,B]$.
\end{lemma}
\begin{proof}
Given a source with distribution $\mu$ and a channel $\nu$, one can define a probability kernel $\nu'$ such that $\mu_\infty \nu' = (\mu\nu)_\infty$ since $\mu_\infty$ is the "input" marginal of $(\mu\nu)_\infty$:
$$
\forall F\in (\mathcal{B}_{A^\mathcal{I}})_\infty, \forall G\in (\mathcal{B}_{B^\mathcal{I}})_\infty, \mu_\infty\nu'(F \times G) = (\mu\nu)_\infty(F\times G) = \mu\nu(F\times G)
$$
Then
$$
\int_F \nu'(x,G) d\mu_\infty = \int_F \nu(x,G) d\mu = \int_F \nu_\infty(x,G) d\mu = \int_F \nu_\infty(x,G) d\mu_\infty
$$
Then $\forall G\in (\mathcal{B}_{B^\mathcal{I}})_\infty, \nu'(x,G) = \nu_\infty(x,G)$ $\mu_\infty$-a.e. 
\end{proof}

\section{Stability Properties of Sources}\label{SectionSourcesStabilityProperties}

Some well known properties of sources or dynamical systems are reminded (cf \cite{Gray09}). Technical lemmas are given for later use.

\begin{definition}
A dynamical system $(\Omega, \mathcal{F}, T,\eta)$ is stationary if 
$$
\forall F\in \mathcal{F}, \eta(T^{-1}F)=\eta(F)
$$
\end{definition} 

\begin{definition}
A dynamical system $(\Omega, \mathcal{F}, T,\eta)$ is recurrent if 
$$
\forall F\in \mathcal{F}, \eta(F \setminus \cup_{k\geq 1} T^{-k}F)=0
$$
i.e. any event $F$ is recurrent.
\end{definition} 

\begin{definition}
A dynamical system $(\Omega, \mathcal{F}, T,\eta)$ is incompressible if 
$$
\forall F\in \mathcal{F} \text{ such that }T^{-1}F \subset F, \eta(F \setminus T^{-1}F)=0
$$
\end{definition} 

A dynamical system is recurrent if and only if it is incompressible (\cite{Gray09}). A stationary dynamical system is recurrent (\cite{Gray09}).

\begin{definition}
A dynamical system $(\Omega, \mathcal{F}, T,\eta)$ is asymptotically mean stationary (AMS) if 
$$
\forall F \in \mathcal{F}, \lim_{n \to \infty}\frac{1}{n} \sum_{k=0}^{n-1}\eta(T^{-k}F)\text{ exists}
$$
If $\overline{\eta}(F)$ is this limit, then $\overline{\eta}$ is a stationary probability on $(\Omega, \mathcal{F})$, called the stationary mean of $\eta$.
\end{definition} 
$\overline{\eta}$ asymptotically dominates $\eta$ ($\eta \ll^a \overline{\eta}$): 
$$
\overline{\eta}(F) =0 \Rightarrow \lim_{n \to \infty} \eta(T^{-n}F)=0
$$
A source with distribution $\mu$ is AMS if and only if  there exists a stationary source with distribution $\eta$ such that $\mu \ll^a \eta$. See \cite{Gray09}. 

The following lemma is Theorem 7.4 of \cite{Gray09}.
\begin{lemma}\label{LemmaTh7dot4Gray2009}
Let  $(\Omega, \mathcal{F}, T,\eta)$ be an AMS dynamical system. Then $\eta$ is dominated by its stationary mean $\overline{\eta}$ ($\eta \ll \overline{\eta}$) if and only if $(\Omega, \mathcal{F}, T,\eta)$ is recurrent.
\end{lemma}
For example, if the shift $T$ is invertible, then $(\Omega, \mathcal{F}, T,\eta)$ is  AMS if and only if $\eta \ll \overline{\eta}$ (\cite{Gray09}). Then an AMS two-sided dynamical system $(\Omega, \mathcal{F}, T,\eta)$ is also recurrent.

\begin{definition}
A dynamical system $(\Omega, \mathcal{F}, T,\eta)$ is ergodic if, for any invariant event $F$ (i.e., $T^{-1}F=F$), either  $\eta(F)=0$ or $\eta(F)=1$.
\end{definition} 

\begin{lemma}\label{LemmaStationaryMeanRestriction}
Let $(\Omega, \mathcal{F}, T, \eta)$ be an AMS dynamical system. $(\Omega, \mathcal{F}_\infty, T, \eta_\infty)$ is an AMS dynamical system and its asymptotic stationary mean is the restriction of the asymptotic stationary mean of $\eta$ to the tail $\sigma$-field:
$$
\overline{(\eta_\infty)}=(\overline{\eta})_\infty
$$
\end{lemma}
\begin{proof}
Let $F \in \mathcal{F}_\infty$. 
\begin{eqnarray}
\lim_{n \to \infty} \frac{1}{n} \sum_{i=0}^{n-1} \eta_\infty(T^{-i}F) & = & \lim_{n \to \infty} \frac{1}{n} \sum_{i=0}^{n-1} \eta(T^{-i}F) \nonumber \\
		& = & \overline{\eta}(F) \nonumber \\
		& = & (\overline{\eta})_\infty(F) \nonumber
\end{eqnarray}
The limit $\overline{(\eta_\infty)}(F)=\lim_{n \to \infty} \frac{1}{n} \sum_{i=0}^{n-1} \eta_\infty(T^{-i}F)$ exists, then $\eta_\infty$ is AMS and $\overline{(\eta_\infty)}=(\overline{\eta})_\infty$
\end{proof}

Lemma \ref{LemmaDominanceOnTailSigmaField} is an extension of some statements of Theorem 3 of \cite{GrayKieffer80} without the assumption of stationarity of the dominating probability. 
\begin{lemma}\label{LemmaDominanceOnTailSigmaField}
Let $[A,X]$ be a source with distribution $\mu$,  let $[A,Y]$ be a source with distribution $\eta$. Let $(\mathcal{B}_{A^\mathcal{I}})_{\infty}$ denote the tail $\sigma$-field $\cap_{n\geq 0}T_A^{-n}\mathcal{B}_{A^\mathcal{I}}$,  $\mu_{\infty}$ and $\eta_{\infty}$ denote the restrictions of $\mu$ and $\eta$ to $(\mathcal{B}_{A^\mathcal{I}})_{\infty}$. Then $\eta$ asymptotically dominates $\mu$ if and only if $\eta_\infty$ dominates $\mu_\infty$:
$$
\mu \ll^a \eta \Leftrightarrow \mu_\infty \ll \eta_\infty
$$
\end{lemma}
\begin{proof}
In the proof of Theorem 3 of \cite{GrayKieffer80}, the stationarity of $\eta$ is used only to prove that $\mu \ll^a \eta \Rightarrow \mu_\infty \ll \eta_\infty$ (Corollary 1 in \cite{GrayKieffer80}). Follows a modified proof of this statement without assuming stationarity of $\eta$.

From the Lebesgue decomposition theorem, for any $n=0,1,2,\cdots$, there exist $B_n \in \mathcal{B}_{A^\mathcal{I}}$ such that for any $F \in \mathcal{B}_{A^\mathcal{I}}$:
\begin{equation}\label{EqLebesgueDecomposition}
\mu T_A^{-n}(F)=\mu T_A^{-n}(F\cap B_n) + \int_F f_n d\eta T_A^{-n}
\end{equation}
where $f_n=\left( \frac{d\mu T_A^{-n}}{d\eta T_A^{-n}}\right)_a$ and $\eta T_A^{-n}(B_n)=0$. Since $\mu \ll^a \eta$, for any $n$, $\mu T_A^{-n} \ll^a \eta T_A^{-n}$. Hence:
$$
\lim_{n \to \infty}\mu(T_A^{-n} B_n)=0
$$
From (\ref{EqLebesgueDecomposition}), for any $F \in \mathcal{B}_{A^\mathcal{I}}$ and any $n$:
$$
0\leq \mu T_A^{-n}(F)- \int_F f_n d\eta T_A^{-n} = \mu T_A^{-n}(F\cap B_n) \leq \mu(T_A^{-n}B_n) \to 0
$$
Thus:
\begin{equation}\label{EqSupDeviation}
\lim_{n \to \infty}\sup_{F \in \mathcal{B}_{A^\mathcal{I}}} \vert \mu T_A^{-n}(F)- \int_F f_n d\eta T_A^{-n} \vert = 0
\end{equation}

Let $F \in (\mathcal{B}_{A^\mathcal{I}})_{\infty}$ such that $\eta(F)=0$. For any $n$, there exists $F_n \in \mathcal{B}_{A^\mathcal{I}}$ such that $F= T_A^{-n}F_n$. By (\ref{EqSupDeviation}):
$$
\lim_{n \to \infty} \vert \mu(F)- \int_{F_n} f_n d\eta T_A^{-n} \vert= \lim_{n \to \infty} \vert \mu T_A^{-n}(F_n)- \int_{F_n} f_n d\eta T_A^{-n} \vert = 0
$$
But $\eta(F)=\eta T_A^{-n}(F_n)=0$, then $\int_{F_n} f_n d\eta T_A^{-n}=0$. Hence:
$$
\mu(F)=0
$$
\end{proof}

\begin{lemma}\label{LemmaAMSProcess}
A dynamical system $(\Omega, \mathcal{F}, T, \eta)$ is AMS if and only if the dynamical system $(\Omega, \mathcal{F}_\infty , T, \eta_\infty)$ is recurrent and AMS.
\end{lemma}

\begin{proof}
$\eta$ is AMS if and only if $\eta \ll^a \overline{\eta}$  if and only if (from Lemma \ref{LemmaDominanceOnTailSigmaField}) $\eta_\infty \ll \overline{\eta}_\infty$ if and only if $\eta_\infty$ is recurrent and AMS  (from Lemma \ref{LemmaTh7dot4Gray2009}).
\end{proof}

\section{Stability Properties of Channels}\label{SectionChannelsStabilityProperties}

Definitions of stability properties (stationarity, quasi-stationarity and a\-symptotically mean stationarity) of channels are reminded. A definition of recurrence for channels is proposed. Dealing with one-sided channels, a clear distinction is made between stationary channels and quasi-stationary channels. The terminology chosen here is different from the one used in \cite{Kakihara04} which names strictly stationary channels for which $\nu(x,T_B^{-1}G) = \nu(T_A x,G) \text{  }\mu \text{.a.e.}$ and stationary those named quasi-stationary here. The choice is made for the sake of consistency with \cite{Gray11}. The quasi-stationary mean of a channel is defined.

\subsection{Stationarity and Quasi-stationarity}

\begin{definition} \label{DefinitionStationaryChannel}
A channel $[A,\nu,B]$ is {\em stationary} with respect to a stationary source $[A,X]$ with distribution $\mu$ if 
$$
\forall G\in \mathcal{B}_{B^\mathcal{I}}, \nu(x,T_B^{-1}G) = \nu(T_A x,G)\text{ }\mu\text{-a.e.}
$$
A channel $[A,\nu,B]$ is stationary if it is stationary with respect to any stationary source.
\end{definition}

This implies that if $\mu$ is stationary and $\nu$ is stationary, $\mu\nu$ is stationary:  $\forall F\in \mathcal{B}_{A^\mathcal{I}}$ and $\forall G\in \mathcal{B}_{B^\mathcal{I}}$
\begin{eqnarray}
\mu\nu(T_A^{-1}F\times T_B^{-1}G) & = & \int_{T_A^{-1}F} \nu(x, T_B^{-1}G) d\mu \nonumber \\
		& = & \int_{T_A^{-1}F} \nu(T_A x, G) d\mu \nonumber \\
		& = & \int_{F} \nu(x, G) d\mu \nonumber \\
		& = & \mu\nu(F \times G) \nonumber
\end{eqnarray}
Since a joint process is stationary if and only if it is stationary on rectangles, $\mu\nu$ is a stationary probability.

\begin{definition} \label{DefinitionQuasiStationaryChannelStationaryChannel}
A channel $[A,\nu,B]$ is {\em quasi-stationary} with respect to a stationary source $[A,X]$ with distribution $\mu$ if the hookup $\mu\nu$ is stationary. A channel $[A,\nu,B]$ is quasi-stationary if it is quasi-stationary with respect to any stationary source. 
\end{definition}

Obviously a stationary channel is quasi-stationary.

\begin{proposition}\label{PropositionQuasiStationaryChannel}
Let $[A,X]$ be a source with distribution $\mu$ and $[A,\nu,B]$ a channel. Then the hookup $\mu\nu$ is stationary if and only if  
\begin{itemize}
	\item $\mu$ is stationary 
	\item and  $\forall F\in \mathcal{B}_{A^\mathcal{I}}$, $\forall G\in \mathcal{B}_{B^\mathcal{I}}$, 
		$$
		\int 1_F(T_A x) \nu(x,T_B^{-1}G) d\mu = \int 1_F(T_A x) \nu(T_A x,G) d\mu
		$$
\end{itemize}
As a consequence
$$
\forall G\in \mathcal{B}_{B^\mathcal{I}}, \nu(x,T_B^{-1}G) = \nu(T_A x,G) \text{  }\mu_{T_A^{-1}\mathcal{B}_{A^\mathcal{I}}}\text{-a.e.}
$$
where $\mu_{T_A^{-1}\mathcal{B}_{A^\mathcal{I}}}$ is the restriction of $\mu$ to the $\sigma$-field $T_A^{-1}\mathcal{B}_{A^\mathcal{I}}$.
\end{proposition}
\begin{proof}
Assume that $\mu\nu$ is stationary. $\mu$ is the (input) marginal of the stationary probability $\mu\nu$ thus $\mu$ is stationary. 

$\forall F\in \mathcal{B}_A$ and $\forall G\in \mathcal{B}_B$, since $\mu\nu$ is stationary:
$$
\mu\nu(T_{AB}^{-1} F\times G)) = \mu\nu(F\times G) 
$$
$$
\int 1_F(T_A x) \nu(x,T_B^{-1}G) d\mu = \int 1_F(x) \nu(x,G) d\mu
$$
Thanks to the stationarity of $\mu$ and to the transfer theorem (or change of variable):
$$
\int 1_F(x) \nu(x,G) d\mu= \int 1_F(x) \nu(x,G) d\mu T_A^{-1} = \int 1_F(T_A x) \nu(T_A x,G) d\mu \nonumber
$$
 Thus 
$$
\int 1_F(T_A x) \nu(x,T_B^{-1}G) d\mu = \int 1_F(T_A x) \nu(T_A x,G) d\mu
$$ 
This proves the necessary condition. 
Sufficiency comes from the same calculations and from the fact that a joint process is stationary if and only if it is stationary on the set of rectangles (\cite{Gray09}).
\end{proof}

If the shift $T_A$ is invertible, then $T_A^{-1}\mathcal{B}_{A^\mathcal{I}}=\mathcal{B}_{A^\mathcal{I}}$. Then quasi-stationarity and stationarity are equivalent for two-sided channels, and this implies the stronger form of Proposition \ref{PropositionQuasiStationaryChannel}  established in \cite{FontanaGrayKieffer81} (Lemma 1):  for invertible shifts, $\mu\nu$ is stationary if and only if both $\mu$ and $\nu$ are stationary.

\subsection{Recurrence}

\begin{definition} \label{DefinitionRecurrentChannel}
A channel $[A,\nu,B]$ is recurrent  with respect to a recurrent source $[A,X]$ with distribution $\mu$ if the hookup $\mu\nu$ is recurrent. A channel is recurrent if it is recurrent w.r.t. any recurrent source.
\end{definition}

\begin{definition} \label{DefinitionIncompressibleChannel}
A channel $[A,\nu,B]$ is incompressible  with respect to a  incompressible  source $\mu$ if the hookup $\mu\nu$ is  incompressible. A channel is incompressible if it is incompressible w.r.t. any incompressible source.
\end{definition}

Recurrence and incompressibility are equivalent properties for sources, the same obviously holds for channels.

\begin{proposition}\label{PropositionRecurrentChannel2}
A channel $[A,\nu,B]$ is recurrent with respect to a  source $[A,X]$ with distribution $\mu$ if and only if $\mu$ is recurrent and $\nu(x,.)$ is recurrent $\mu$-a.e.
\end{proposition}
The proof relies on the following lemmas.
\begin{lemma}\label{LemmaRecurrentChannel1}
Let $O_x=\{ y \in B^\mathcal{I}/ (x,y)\in O \}$ be the section of $O$ at $x$. A channel $[A,\nu,B]$ is recurrent with respect to $\mu$ if and only if 
$$
\forall O \in \mathcal{B}_{A^\mathcal{I}\times B^\mathcal{I}}, \nu(x, O_x \setminus \bigcup_{i=1}^\infty (T_{AB}^{-i}O)_x) = 0 \text{  }\mu\text{-a.e.}
$$
\end{lemma}
\begin{proof}[Proof of Lemma \ref{LemmaRecurrentChannel1}]
Since the section of an union is the union of sections and since the section of a set-difference is the set-difference of the sections:
\begin{eqnarray}
   \mu\nu(O \setminus \bigcup_{i=1}^\infty (T_{AB}^{-i}O)) & = & 0 \nonumber \\
\Leftrightarrow \int \nu(x,(O \setminus \bigcup_{i=1}^\infty (T_{AB}^{-i}O))_x) d\mu & = & 0 \nonumber \\
\Leftrightarrow \nu(x,(O \setminus \bigcup_{i=1}^\infty (T_{AB}^{-i}O))_x) d\mu & = & 0 \text{  }\mu\text{.-a.e.}\nonumber \\
\Leftrightarrow \nu(x,O_x \setminus \bigcup_{i=1}^\infty (T_{AB}^{-i}O)_x)d\mu & = & 0 \text{  }\mu\text{-a.e.}\nonumber 
\end{eqnarray}
\end{proof}

\begin{lemma}\label{LemmaRecurrenceOnRectangles}
Let $\mathcal{G}$ be the set of rectangles $F \times G$ of $A^\mathcal{I}\times B^\mathcal{I}$, $\mathcal{F}$ the field generated by $\mathcal{G}$. Let $\eta$ be a probability on $(A^\mathcal{I}\times B^\mathcal{I}, \mathcal{B}_{A^\mathcal{I}\times B^\mathcal{I}})$. Then the following statements are equivalent:
\begin{enumerate}
	\item \label{ItemRecurrenceOnGeneratingSets} for any $F\times G \in \mathcal{G}$, $\eta(F\times G \setminus \cup_{k\geq 1}T_{AB}^{-k}F\times G)=0$.
	\item \label{ItemRecurrenceOnCountableUnionsOfFieldSets} For any countable family $(R_i)_{i\geq0}$ of elements of the field $\mathcal{F}$, $\eta( (\cup_{i\geq0} R_i) \setminus \cup_{k\geq 1}T_{AB}^{-k}(\cup_{i\geq0} R_i))=0$.
	\item \label{ItemRecurrence} The dynamical system $(A^\mathcal{I}\times B^\mathcal{I}, \mathcal{B}_{A^\mathcal{I}\times B^\mathcal{I}}, T_{AB} , \eta)$ is recurrent.
	\item \label{ItemIncompressibilityOnGeneratingSets} For any $F\times G \in \mathcal{G}$ such that $T_{AB}^{-1}F\times G \subset F\times G$, $\eta(F\times G \setminus T_{AB}^{-1}F\times G)=0$.
	\item \label{ItemIncompressibilityOnCountableUnionsOfFieldSets} Let $(R_i)_{i\geq0}$ a countable family of elements of the field $\mathcal{F}$ such that $\forall i$ $T_{AB}^{-1}R_i \subset R_i$, then $\eta( (\cup_{i\geq0} R_i) \setminus T_{AB}^{-1}(\cup_{i\geq0} R_i))=0$ 
	\item \label{ItemIncompressibility} The dynamical system $(A^\mathcal{I}\times B^\mathcal{I}, \mathcal{B}_{A^\mathcal{I}\times B^\mathcal{I}}, T_{AB} , \eta )$ is incompressible.
\end{enumerate}
\end{lemma}
\begin{proof}[Proof of Lemma \ref{LemmaRecurrenceOnRectangles}]
See \ref{SectionProofsOfLemmas}.
\end{proof}
\begin{proof}[Proof of Proposition \ref{PropositionRecurrentChannel2}]
Assume that $\mu\nu$ is recurrent. Let $F\in \mathcal{B}_{A^\mathcal{I}}$ such that $T_A^{-1}F \subset F$. Then $T_A^{-1}F \times T_B^{-1}B^\mathcal{I }\subset F\times B^\mathcal{I }$. $\mu\nu$ is recurrent, equivalently incompressible. Then:
$$
\mu(T_A^{-1}F)=\mu\nu(T_A^{-1}F \times T_B^{-1}B^\mathcal{I })=\mu\nu( F\times B^\mathcal{I })=\mu(F)
$$
Thus $\mu$ is incompressible, equivalently recurrent.

By Lemma \ref{LemmaRecurrentChannel1} applied to sets $O=A^\mathcal{I}\times G$ , $\nu(x,.)$ is recurrent $\mu$-a.e.

Assume now that $\mu$ is recurrent and $\nu(x,.)$ is recurrent $\mu$-a.e. From Lemma \ref{LemmaRecurrenceOnRectangles}, it is enough to prove recurrence or incompressibility for rectangles. For a rectangle $F \times G \in \mathcal{B}_{A^\mathcal{I}} \times \mathcal{B}_{B^\mathcal{I}}$ such that $T_A^{-1}F \subset F$ and $T_B^{-1}G \subset G$:
$$
\mu\nu(T_A^{-1}F \times T_B^{-1}G) =\int_{T_A^{-1}F} \nu(x, T_B^{-1}G) d\mu
$$
$\nu(x,.)$ is recurrent thus incompressible  $\mu$-a.e. then:
$$
\mu\nu(T_A^{-1}F \times T_B^{-1}G) =\int_{T_A^{-1}F} \nu(x, G) d\mu
$$
$T_A^{-1}F\subset F$ and, since $\mu$ is recurrent, $\mu(T_A^{-1}F)=\mu(F)$ then:
$$
\mu\nu(T_A^{-1}F \times T_B^{-1}G) =\int_{F} \nu(x, G) d\mu = \mu\nu(F\times G)
$$

\end{proof}

\subsection{Asymptotically Mean Stationarity}

\begin{definition}\label{DefinitionAMSChannel}
A channel $[A,\nu,B]$ is  AMS with respect to an AMS  source $[A,X]$ with distribution $\mu$ if the hookup $\mu\nu$ is  AMS. A channel is AMS if it is AMS with respect to any AMS source.
\end{definition}

A detailed analysis of AMS channels is proposed in section \ref{SectionOneSidedAMSChannels}. This analysis relies on properties of channels which are both recurrent and AMS, studied in section \ref{SectionRAMSChannels}. Such channels will be called R-AMS. Recurrent and AMS sources will also be called R-AMS.

\begin{definition}\label{DefinitionRAMSChannel}
A channel $[A,\nu,B]$ is  recurrent and AMS (R-AMS) with respect to a recurrent and AMS (R-AMS) source $[A,X]$ with distribution $\mu$ if the hookup $\mu\nu$ is  recurrent and AMS (R-AMS). A channel is R-AMS if it is R-AMS with respect to any R-AMS source.
\end{definition}

Let $\mu$ be the distribution of an AMS source with stationary mean $\overline{\mu}$: $\mu\ll^a \overline{\mu}$. Let $\nu$ be an AMS channel. Then $\mu\nu$ is AMS: $\mu\nu \ll^a \overline{\mu\nu}$. The "input marginal" of $\overline{\mu\nu}$ is $\overline{\mu}$, hence there exists a (unique modulo $\overline{\mu}$) channel $\overline{\nu}_\mu$ quasi-stationary w.r.t $\overline{\mu}$ such that $\overline{\mu\nu}=\overline{\mu}\ \overline{\nu}_\mu$. 

\begin{definition}
Let $[A,\nu,B]$ be an AMS channel and $[A,X]$ an AMS source with distribution $\mu$. The channel $[A,\overline{\nu}_{\mu},B]$ such that $\overline{\mu\nu} = \overline{\mu}\ \overline{\nu}_\mu$ is the quasi-stationary mean of the AMS channel $\nu$ with respect to the AMS source distribution $\mu$.
\end{definition}

{\bf Remark:}
\begin{itemize}
	\item Let $\mu$ be the distribution of  an R-AMS source with stationary mean $\overline{\mu}$ and $\nu$  an R-AMS channel. By Lemma \ref{LemmaTh7dot4Gray2009}, $\mu\ll\overline{\mu}$ and $\mu\nu \ll \overline{\mu\nu}$. In this case, the quasi-stationary mean $\overline{\nu}_\mu$ of $\nu$ with respect to $\mu$ is such that  $\mu\nu \ll \overline{\mu}\ \overline{\nu}_\mu$. 
	\item  given an R-AMS source distribution $\mu$ and an R-AMS channel $\nu$, from Lemma \ref{Lemma2Fontana} given below, $\mu \ll \overline{\mu} \Rightarrow \mu\nu \ll \overline{\mu}\nu$ and since $\nu$ is R-AMS $\mu\nu \ll \overline{\mu}\nu \ll \overline{\overline{\mu}\nu}= \overline{\mu}\ \overline{\nu_{\overline{\mu}}}$. According to \cite{FontanaGrayKieffer81}, in general $\overline{\nu}_\mu \neq \ \overline{\nu_{\overline{\mu}}}$. It will be shown in section~\ref{SectionQuasiStationaryMeanErgodicRAMSChannel} that $\overline{\nu}_\mu = \ \overline{\nu_{\overline{\mu}}}$ when considering ergodic R-AMS sources and ergodic R-AMS channels.
\end{itemize}

\section{Classification of channels}\label{SectionClassificationOfChannels}

It is possible to hierarchically classify channels according to stability properties. It has been mentioned above that stationary channels are quasi-stationary. It is proved below that quasi-stationary channels are R-AMS and R-AMS channels are AMS.

\begin{proposition}\label{PropositionChannelClassification}
Let $[A,\nu,B]$ be a channel and $[A,X]$ a source with distribution $\mu$.
\begin{enumerate}
	\item \label{PropositionChannelClassificationFirstItem} if $\nu$ is quasi-stationary and $\mu$ is R-AMS then $\mu\nu$ is R-AMS.
	\item \label{PropositionChannelClassificationSecondItem}if $\nu$ is R-AMS and $\mu$ is AMS then $\mu\nu$ is AMS.
	\item \label{PropositionChannelClassificationThirdItem}The set of stationary channels is a subset of quasi-stationary channels which is a subset of R-AMS channels which is a subset of AMS channels.
\end{enumerate}
\end{proposition}
The proof of Proposition \ref{PropositionChannelClassification} relies on the following lemma which is Lemma 2 of \cite{FontanaGrayKieffer81}.
\begin{lemma}\label{Lemma2Fontana}
Let $\mu$ and $\eta$ be the distributions of two sources on the same alphabet $A$. Then,  for any channel $[A,\nu,B]$ 
$$
\mu \ll \eta \Rightarrow \mu\nu \ll \eta\nu
$$
\end{lemma}

\begin{proof}[Proof of Proposition \ref{PropositionChannelClassification}]
~\\
\begin{description}
	\item[\ref{PropositionChannelClassificationFirstItem})] Let $\nu$ be a quasi-stationary channel and $\mu$ the distribution of an R-AMS source. Then $\mu \ll \overline{\mu}$. By Lemma \ref{Lemma2Fontana}, this implies that $\mu\nu \ll \overline{\mu}\nu$. $\mu\nu$ is dominated by the stationary probability $\overline{\mu}\nu$ then, by Lemma \ref{LemmaTh7dot4Gray2009}, $\mu\nu$ is R-AMS.
	\item[\ref{PropositionChannelClassificationSecondItem})] Let $\nu$ be an R-AMS channel and $\mu$  the distribution of an AMS source. Then $\mu \ll^a \overline{\mu}$. By Lemma \ref{LemmaAMSProcess}, $\mu_\infty \ll \overline{\mu}_\infty$. By Lemma \ref{Lemma2Fontana}, this implies that $\mu_\infty\nu_\infty \ll \overline{\mu}_\infty\nu_\infty$. $ \overline{\mu}_\infty\nu_\infty \ll \overline{\mu}_\infty\ \overline{\nu}_{\overline{\mu}_\infty}$ since $\nu$ is R-AMS and $\overline{\mu}$ stationary. Then $\mu\nu \ll^a\overline{\mu}\ \overline{\nu}_{\overline{\mu}}$,  then $\mu\nu$ is AMS.
	\item [\ref{PropositionChannelClassificationThirdItem})] This is a direct consequence of (\ref{PropositionChannelClassificationFirstItem}),  (\ref{PropositionChannelClassificationSecondItem}) and the fact that a stationary channel is quasi-stationary.
\end{description}
\end{proof}

\section{R-AMS channels}\label{SectionRAMSChannels}

In this section, necessary and sufficient conditions for a channel to be R-AMS are given: the R-AMS property of a channel needs to be checked only on stationary sources and a channel is R-AMS if and only if it is dominated by a quasi-stationary channel. A sufficient condition is also proved: if a channel is made of a collection of probabilities $\nu(x,.)$, each R-AMS $\mu$-a.s. for a stationary $\mu$, then the channel is R-AMS w.r.t. $\mu$

\begin{proposition}\label{PropositionNecessarySufficientConditionRAMSChannel}
The following statements are equivalent:
\begin{enumerate}
	\item \label{PropositionNSCondRAMSChlFirstStatement} the channel $[A,\nu,B]$ is R-AMS
	\item \label{PropositionNSCondRAMSChlSecondStatement} the channel $[A,\nu,B]$ is R-AMS with respect to any stationary source
	\item \label{PropositionNSCondRAMSChlThirdStatement} for any stationary source $[A,X]$ with distribution $\mu$ there exists a quasi-stationary channel $\overline{\nu}_\mu$ such that $\nu(x,.) \ll \overline{\nu}_{\mu}(x,.)$ $\mu$-a.e.
	\item \label{PropositionNSCondRAMSChlFourthStatement} for any R-AMS source $[A,X]$ with distribution $\mu$ there exists a quasi-stationary channel $\overline{\nu}_\mu$ such that $\nu(x,.) \ll \overline{\nu}_{\mu}(x,.)$ $\overline{\mu}$-a.e. and  $\mu$-a.e	
\end{enumerate}
\end{proposition}
The proof of Proposition \ref{PropositionNecessarySufficientConditionRAMSChannel} uses the following  lemma which gathers Lemma 3 and Lemma 4 of \cite{FontanaGrayKieffer81}. It states that "hookup dominance" is equivalent to "channel dominance".
\begin{lemma}\label{LemmaChannelDominanceEquivHookupDominance}
Let $[A,X]$ be a source with distribution $\mu$. Let  $[A,\nu,B]$ and $[A,\nu',B]$ be two arbitrary channels. Then 
$$
\nu(x,.) \ll \nu'(x,.) \text{ }\mu\text{-a.e. } \Leftrightarrow \mu\nu \ll \mu\nu'
$$
\end{lemma}

\begin{proof}[Proof of Proposition \ref{PropositionNecessarySufficientConditionRAMSChannel}]
~\\
\begin{description}
	\item[(\ref{PropositionNSCondRAMSChlSecondStatement})$\Leftrightarrow$(\ref{PropositionNSCondRAMSChlFirstStatement})]  Assume that a  channel $[A,\nu,B]$ is such that $\eta\nu$ is R-AMS for any distribution  $\eta$ of a stationary source. 

Let $[A,X]$  be an R-AMS source with distribution $\mu$. Then the stationary mean $\overline{\mu}$ exists and dominates $\mu$: $\mu \ll \overline{\mu}$.  By Lemma \ref{Lemma2Fontana}, $\mu\nu \ll \overline{\mu}\nu$. 

$\overline{\mu}$ is stationary, then, by assumption, $\overline{\mu}\nu$ is R-AMS. Then $\overline{\mu}\nu \ll \overline{\overline{\mu}\nu}$ where $\overline{\overline{\mu}\nu}$ is stationary. Since $\ll$ is transitive, $\mu\nu \ll \overline{\overline{\mu}\nu}$, thus $\mu\nu$ is R-AMS. 

Obviously if $\nu$ is R-AMS w.r.t. any R-AMS source, it is R-AMS w.r.t. any stationary source (a stationary source is R-AMS).

	\item[(\ref{PropositionNSCondRAMSChlThirdStatement})$\Leftrightarrow$(\ref{PropositionNSCondRAMSChlSecondStatement})] From Lemma \ref{LemmaChannelDominanceEquivHookupDominance}, (\ref{PropositionNSCondRAMSChlThirdStatement})$\Leftrightarrow \mu\nu \ll \mu\overline{\nu}_\mu$ where $\mu\overline{\nu}_\mu$ is stationary, for any stationary $\mu$.

	\item[(\ref{PropositionNSCondRAMSChlFourthStatement})$\Leftrightarrow$(\ref{PropositionNSCondRAMSChlThirdStatement})] It is obvious that \ref{PropositionNSCondRAMSChlFourthStatement})$\Rightarrow$(\ref{PropositionNSCondRAMSChlThirdStatement}). Assume (\ref{PropositionNSCondRAMSChlThirdStatement}). Let $[A,X]$ be an R-AMS source with distribution $\mu$.  Since $\overline{\mu}$ is stationary, by  (\ref{PropositionNSCondRAMSChlThirdStatement}), there exists a quasi-stationary channel $\overline{\nu}_{\overline{\mu}}$ such that that $\nu(x,.) \ll \overline{\nu}_{\overline{\mu}}(x,.)$ $\overline{\mu}$-a.e. and thus $\mu$.a.e. 

\end{description} 
\end{proof}

\begin{proposition}\label{PropositionKernelRAMSisRAMSChannel}
Let $[A,X]$ be a stationary source with distribution $\mu$ and let $[A,\nu,B]$ be a channel such that $(B^\mathcal{I},\mathcal{B}_{B^\mathcal{I}},\nu(x,.), T_B)$ is an R-AMS dynamical system $\mu$-a.e. Then the channel  $\nu$ is  R-AMS w.r.t. $\mu$.
\end{proposition}
\begin{proof}
Let $\mu$ be the distribution of a stationary source. 
Let $\Omega_1$ be the set of $x$'s such that $(B^\mathcal{I},\mathcal{B}_{B^\mathcal{I}},\nu(x,.), T_B)$ is a recurrent and AMS dynamical system. $\mu(\Omega_1)=1$.

For any $x\in \Omega_1$ and any $G$, since $\nu(x,.)$ is R-AMS, the limit 
$$
\overline{\nu}(x,G)=\lim_{n \to \infty} \frac{1}{n}\sum_{i=0}^{n-1} \nu(x,T_B^{-i}G)
$$
 exists, $\overline{\nu}(x,.)$ is a stationary probability and $\nu(x,.) \ll \overline{\nu}(x,.)$.

Let $\phi: A^\mathcal{I}\times \mathcal{B}_{B^\mathcal{I}}  \to [0,1]$ be such that for any $G$,  $\forall x\in \Omega_1$, $\phi(x,G)=\overline{\nu}(x,G)$ and $\forall x\in A^\mathcal{I} \setminus \Omega_1$, $\phi(x,G)=\nu(x,G)$. 
For any $x$, the set function $G \mapsto \phi(x,G)$ is a probability on $(B^\mathcal{I}, \mathcal{B}_{B^\mathcal{I}})$. Moreover the function $x \mapsto \phi(x,G)$ is measurable for any $G$. In other words $\phi$ is a channel (probability kernel). 

The function $x\mapsto \phi(x,G)$ is $\mu$-integrable for any $G$, then, thanks to the pointwise ergodic theorem (see e.g. \cite{Kakihara99}), there exits a $\mu$-integrable function $\psi_G$, for any $G$, such that:
\begin{enumerate}
	\item \label{PWETInvariantFunction}$\psi_G(T_A x)= \psi_G(x)$  $\mu$-a.e.
	\item $\psi_G(x) =\lim_{n\to \infty}  \frac{1}{n}\sum_{i=0}^{n-1} \phi((T_A^i x, G)$ $\mu$-a.e. Since $\mu(\Omega_1)=1$, $\psi_G(x)=  \lim_{n\to \infty}  \frac{1}{n}\sum_{i=0}^{n-1} \overline{\nu}(T_A^ix,G)$ $\mu$-a.e.
	\item \label{PWETInvariantSet} for any invariant $F$, $\int_F \psi_G(x)d\mu = \int_F \phi(x,G) d\mu$ and, since $\mu(\Omega_1)=1$, $\int_F \psi_G(x)d\mu= \int_F \overline{\nu}(x,G)d\mu $
\end{enumerate}
From the Vitali-Hahn-Saks theorem, the set function $\Psi(x,.): G \mapsto \psi_G(x)$ is a probability $\mu$-a.e. Moreover the function $x\mapsto \psi_G(x)=\Psi(x,G)$ is measurable for any $G$. This means that $\Psi=\{ \Psi(x,.), x\in A^\mathcal{I}\}$ is a channel.
$$
\mu\Psi(T_A^{-1}F\times T_B^{-1}G)  =  \int_{T_A^{-1}F} \psi_{T_B^{-1}G}(x) d\mu
$$
By (\ref{PWETInvariantFunction})
$$
\mu\Psi(T_A^{-1}F\times T_B^{-1}G)  =  \int_{T_A^{-1}F} \psi_{T_B^{-1}G}(T_Ax) d\mu
$$
$\mu$ is stationary then
$$
\mu\Psi(T_A^{-1}F\times T_B^{-1}G)  =  \int_{F} \psi_{T_B^{-1}G}(x) d\mu
$$
$\overline{\nu}(x,.)$ is stationary then
$$
\mu\Psi(T_A^{-1}F\times T_B^{-1}G)  =  \int_{F} \psi_{G}(x) d\mu
$$
$$
\mu\Psi(T_A^{-1}F\times T_B^{-1}G)  =  \mu\Psi(F\times G) \nonumber
$$
Thus $\mu\Psi$ is stationary.

Let $O$ such that $\mu\Psi(O)=\int_{A^\mathcal{I}} \psi_{O_x}(x) d\mu = 0$. Then, by~\ref{PWETInvariantSet}), $A^\mathcal{I}$ being an invariant set:
$$
\int_{A^\mathcal{I}} \overline{\nu}(x,O_x) d\mu = \int_{A^\mathcal{I}} \psi_{O_x}(x) d\mu = 0
$$
Then $\overline{\nu}(x,O_x) =0$ $\mu$-a.e. Since $\nu(x,.) \ll \overline{\nu}(x,.)$ ($\nu(x,.)$ is R-AMS)
$$
\mu\nu(O)=\int \nu(x,O_x) d\mu = 0
$$
Then $\mu\nu \ll \mu\Psi$ which is stationary. This implies that $\mu\nu$ is R-AMS. 
 \end{proof}

\section{One-sided AMS channels}\label{SectionOneSidedAMSChannels}

In this section, necessary and sufficient conditions for a channel to be AMS are given: the AMS property of a channel needs to be checked only on stationary sources and a channel is AMS if and only if it is asymptotically dominated by a quasi-stationary channel. A sufficient condition is also proved: if a channel is made of a collection of probabilities $\nu(x,.)$, each AMS $\mu$-a.s. for a stationary $\mu$, then the channel is AMS w.r.t. $\mu$.

\begin{proposition}\label{PropositionNecessarySufficientConditionAMSChannel2}
The following statements are equivalent:
\begin{enumerate}
	\item \label{PropositionNSCondAMSChl2FirstStatement} the channel $[A,\nu,B]$ is AMS
	\item \label{PropositionNSCondAMSChl2SecondStatement} the channel restriction $[A,\nu_\infty,B]$ is R-AMS
	\item \label{PropositionNSCondAMSChl2ThirdStatement} the channel $[A,\nu,B]$ is AMS with respect to any stationary source
	\item \label{PropositionNSCondAMSChl2FourthStatement} for any stationary source $[A,X]$ with distribution $\mu$ there exists a quasi-stationary channel $\overline{\nu}_\mu$ such that $\nu(x,.) \ll^a \overline{\nu}_{\mu}(x,.)$ $\mu_\infty$-a.e.
	\item \label{PropositionNSCondAMSChl2FifthStatement} for any AMS source $[A,X]$ with distribution $\mu$ there exists a quasi-stationary channel $\overline{\nu}_\mu$ such that $\nu(x,.) \ll^a \overline{\nu}_{\mu}(x,.)$ $\overline{\mu}_\infty$-a.e.
\end{enumerate}
\end{proposition}
\begin{proof}
~\\
\begin{description}
	\item[(\ref{PropositionNSCondAMSChl2SecondStatement})$\Leftrightarrow$(\ref{PropositionNSCondAMSChl2FirstStatement})] 

Assume that $\nu_\infty$ is R-AMS and let $[A,X]$ be an AMS source with distribution $\mu$. By Lemma \ref{LemmaAMSProcess}, $\mu_\infty$ is R-AMS and thus $\mu_\infty\nu_\infty$ is R-AMS.  By Lemma \ref{LemmaAMSProcess}, $\mu\nu$ is AMS.

Assume now that $\nu$ is AMS and let $\mu$ be the distribution of a source such that $\mu_\infty$ is an R-AMS probability on $(A^\mathcal{I}, (\mathcal{B_{A^\mathcal{I}})_\infty})$. 

$\mu_\infty$ is R-AMS then by Lemma \ref{LemmaAMSProcess} $\mu$ is AMS. $\mu$ and $\nu$ are AMS then $\mu\nu$ is AMS and, by Lemma \ref{LemmaAMSProcess}, $(\mu\nu)_\infty=\mu_\infty\nu_\infty$ is R-AMS.

	\item[(\ref{PropositionNSCondAMSChl2ThirdStatement})$\Leftrightarrow$(\ref{PropositionNSCondAMSChl2FirstStatement})] 

It is obvious that \ref{PropositionNSCondAMSChl2FirstStatement})$\Rightarrow$(\ref{PropositionNSCondAMSChl2ThirdStatement}). 

Assume (\ref{PropositionNSCondAMSChl2ThirdStatement}). Let $[A,X]$ be an AMS source with distribution $\mu$.  Since $\overline{\mu}$ is stationary, $\overline{\mu}\nu$ is AMS. Thus, by Lemma \ref{Lemma2Fontana} $\mu_\infty\nu_\infty \ll \overline{\mu}_\infty\nu_\infty$ and by Lemma \ref{LemmaAMSProcess} $\overline{\mu}_\infty\nu_\infty$ is R-AMS then
$$
\mu_\infty\nu_\infty \ll \overline{\mu}_\infty\nu_\infty \ll \overline{\mu}_\infty\ \overline{\nu_\infty}
$$
Then, since $(\overline{\nu})_\infty = \overline{\nu_\infty}$,  by Lemma \ref{LemmaAMSProcess} $\mu\nu \ll^a \overline{\mu}\ \overline{\nu}$.

	\item[(\ref{PropositionNSCondAMSChl2FourthStatement})$\Leftrightarrow$(\ref{PropositionNSCondAMSChl2ThirdStatement})] 

Let $\nu$ be such that (\ref{PropositionNSCondAMSChl2FourthStatement}) holds and let $[A,X]$ be a stationary source with distribution $\mu$: 
$$
\nu(x,.) \ll^a \overline{\nu}_\mu(x,.) \text{  }\mu_\infty\text{-a.e.}
$$
Then, by  Lemma \ref{LemmaDominanceOnTailSigmaField}, 
$$
\nu_\infty(x,.) \ll (\overline{\nu}_\mu)_\infty(x,.)\text{  }\mu_\infty\text{-a.e.}
$$
Thanks to Lemma \ref{LemmaChannelDominanceEquivHookupDominance}, this implies that $(\mu\nu)_\infty \ll (\mu\overline{\nu}_\mu)_\infty$. Then, by  Lemma \ref{LemmaAMSProcess}, $\mu\nu \ll^a \mu\overline{\nu}_\mu$. Hence $\nu$ is AMS w.r.t. $\mu$.

Let $\nu$ such that (\ref{PropositionNSCondAMSChl2ThirdStatement}) holds and let $[A,X]$ be a stationary source with distribution $\mu$. $\mu\nu$ is AMS then 
$$
\mu_\infty\nu_\infty \ll \mu_\infty\overline{\nu}_\mu
$$
Thanks to Lemma \ref{LemmaChannelDominanceEquivHookupDominance}, it holds that 
$$
\nu_\infty(x, .) \ll (\overline{\nu}_\mu(x,.))_\infty \text{  }\mu_\infty\text{-a.e.}
$$
 By Lemma  \ref{LemmaDominanceOnTailSigmaField}, $\nu(x, .) \ll \overline{\nu}_\mu(x,.)$ $\mu_\infty$-a.e.

	\item[(\ref{PropositionNSCondAMSChl2FifthStatement})$\Leftrightarrow$(\ref{PropositionNSCondAMSChl2FourthStatement})] 

It is obvious that (\ref{PropositionNSCondAMSChl2FifthStatement}) implies (\ref{PropositionNSCondAMSChl2FourthStatement}). 

Assume (\ref{PropositionNSCondAMSChl2FourthStatement}) and let $[A,X]$ be an AMS source with distribution $\mu$. $\overline{\mu}$ exists and is stationary then there exists a quasi-stationary channel $\overline{\nu}_{\overline{\mu}}$  such that $\nu(x,.) \ll^a \overline{\nu}_{\mu}(x,.)$ $\overline{\mu}_\infty$-a.e.
\end{description} 
\end{proof}

\begin{proposition}\label{PropositionKernelAMSisAMSChannel}
Let $[A,X]$ be a stationary  source with distribution $\mu$ and let $[A,\nu,B]$ be a channel such that $(B^\mathcal{I},\mathcal{B}_{B^\mathcal{I}},\nu(x,.), T_B)$ is an AMS dynamical system $\mu$-a.e. Then the channel  $\nu$ is  AMS w.r.t. $\mu$.
\end{proposition}
\begin{proof}
Let $\mu$ be the distribution of a stationary source and let $[A,\nu,B]$ be a channel such that $(B^\mathcal{I},\mathcal{B}_{B^\mathcal{I}},\nu(x,.), T_B)$ is an AMS dynamical system $\mu$-a.e. By Lemma \ref{LemmaAMSProcess}, $(B^\mathcal{I},(\mathcal{B}_{B^\mathcal{I}})_\infty,(\nu(x,.))_\infty, T_B)$ is an R-AMS dynamical system $\mu$-a.e. By Proposition  \ref{PropositionKernelRAMSisRAMSChannel}, the channel restriction $\nu_\infty$ is R-AMS w.r.t $\mu_\infty$. Thanks to Proposition \ref{PropositionNecessarySufficientConditionAMSChannel2}, $\nu$ is AMS w.r.t. $\mu$.
\end{proof}
Markov channels  have been shown to be AMS in \cite{KiefferRahe81} following a quite involved proof. A Markov channel $[A,\nu,B]$ is a family of finite non-homogene\-ous Markov processes with distributions $\nu(x,.), x \in A^\mathcal{I}$. From \cite{FaigleSchonuth07}, a finite non-homogeneous Markov process is AMS. Thus, by Proposition \ref{PropositionKernelAMSisAMSChannel} and Proposition \ref{PropositionNecessarySufficientConditionAMSChannel2}, a Markov channel is AMS.

\section{Cascades of Channels}\label{SectionCascadesOfChannels}

Cascades of channels arise, for example, when considering a communication system which is a sequence or cascade made of a coder (deterministic channel), of a noisy communication channel and of a decoder (deterministic channel). Another example is that of multi-hop communications through a sequence of noisy channels. 

For ease of reading, channels will be denoted by conditional probabilities: 
$$
\nu(x,.) = P_{Y|X}(.|x)
$$
Two channels $[A,P_{Y|X},B]$ and $[B,P_{Z|Y},C]$ form a cascade if, given a source $[A,X]$, $X \rightarrow Y \rightarrow Z$ is a Markov chain, where $Y$ and $Z$ are respectively the processes corresponding to the output marginals of $P_{XY}$ and $P_{YZ}$:
$$
\forall H \in \mathcal{B}_{C^\mathcal{I}}, P_{Z|XY}(H|xy)=P_{Z|Y}(H|y) \text{  }P_{XY}\text{-a.e.}
$$
In this case (cf \cite{Gray11}):
$$
\forall H \in \mathcal{B}_{C^\mathcal{I}}, P_{Z|X}(H|x)=\int_{B^\mathcal{I}} P_{Z|Y}(H|y) dP_{Y|X}
$$
The cascading operation on channels is in fact a Markovian composition.

The propositions given below show that the classes of stationary channels, of quasi-stationary channels, of R-AMS channels and of AMS channels are stable for cascading or Markovian composition. It is shown that if a cascade ends with a recurrent channel then the cascade is recurrent.

The following proposition is proved in \cite{Gray11}.
\begin{proposition}\label{PropositionCascadeOfStationaryChannels}
The cascade of two stationary channels $[A,P_{Y|X},B]$ and $[B,P_{Z|Y},C]$ is a stationary channel $[A,P_{Z|X},C]$. 
\end{proposition}

The same holds for quasi-stationary channels.

\begin{proposition}\label{PropositionCascadeOfQuasiStationaryChannels}
The cascade of   two quasi-stationary channels $[A,P_{Y|X},B]$ and $[B,P_{Z|Y},C]$ is a quasi-stationary channel $[A,P_{Z|X},C]$. 
\end{proposition}
\begin{proof}
Let $P_X$ be the distribution of a stationary source on the alphabet $A$. $P_{XY}$ and $P_Y$  respectively denote the distribution of the hookup of the source $[A,X]$ and of the channel $[A,P_{Y|X},B]$ and  the output marginal. The channels are quasi-stationary then $P_{XY}$, $P_Y$, $P_{YZ}$ and $P_Z$ are stationary.

 For any $F\in \mathcal{B}_{A^\mathcal{I}}$, any $G\in \mathcal{B}_{B^\mathcal{I}}$ and any $H\in \mathcal{B}_{C^\mathcal{I}}$:
\begin{multline}
P_{Y}(T_B^{-1}G).P_{XYZ}(T_A^{-1}F\times T_B^{-1}G \times T_C^{-1}H)  = \nonumber \\
					 \int_{T_B^{-1}G} P_{XYZ}(T_A^{-1}F\times T_B^{-1}G \times T_C^{-1}H) dP_Y = \nonumber \\
					  \int_{T_B^{-1}G} \int_{T_A^{-1}F\times T_B^{-1}G} P_{Z|XY}(T_C^{-1}H|xy) dP_{XY} dP_Y = \nonumber \\
					  \int_{T_B^{-1}G} \int_{T_A^{-1}F\times T_B^{-1}G} P_{Z|Y}(T_C^{-1}H|y) dP_{XY} dP_Y\nonumber
\end{multline}
The last step is due to the fact that $X\rightarrow Y \rightarrow Z$ is a Markov chain ($P_{Z|XY}(T_C^{-1}H|xy)=P_{Z|Y}(T_C^{-1}H|y)$). Thanks to Fubini's theorem:
\begin{multline}
P_{Y}(T_B^{-1}G).P_{XYZ}(T_A^{-1}F\times T_B^{-1}G \times T_C^{-1}H) =   \\
			 \int_{T_A^{-1}F\times T_B^{-1}G} \int_{T_B^{-1}G} P_{Z|Y}(T_C^{-1}H|y) dP_Y dP_{XY} \nonumber
\end{multline}
The channel $[B,P_{Z|Y},C]$ is quasi-stationary then, from Proposition \ref{PropositionQuasiStationaryChannel}:
$$
\int_{T_B^{-1}G} P_{Z|Y}(T_C^{-1}H|y) dP_Y= \int_{T_B^{-1}G} P_{Z|Y}(H|T_B y) dP_Y
$$
Then
\begin{multline}
P_{Y}(T_B^{-1}G).P_{XYZ}(T_A^{-1}F\times T_B^{-1}G \times T_C^{-1}H)  =  \\
			   \int_{T_A^{-1}F\times T_B^{-1}G} \int_{T_B^{-1}G} P_{Z|Y}(H|T_B y) dP_Y dP_{XY} = \nonumber \\
					\int_{T_B^{-1}G} \int_{T_A^{-1}F\times T_B^{-1}G} P_{Z|Y}(H|T_B y) dP_{XY} dP_Y \nonumber
\end{multline}
The last step holds thanks to Fubini's theorem. $X\rightarrow Y \rightarrow Z$ is a Markov chain and $P_{XY}$ is stationary, hence:
\begin{multline}
P_{Y}(T_B^{-1}G).P_{XYZ}(T_A^{-1}F\times T_B^{-1}G \times T_C^{-1}H)  = \\
		 \int_{T_B^{-1}G} \int_{T_A^{-1}F\times T_B^{-1}G} P_{Z|XY}(H|T_A x T_B y) dP_{XY} dP_Y = \nonumber \\
					\int_{T_B^{-1}G} \int_{F\times G} P_{Z|XY}(H|xy) dP_{XY} dP_Y = \nonumber \\
						 P_{Y}(T_B^{-1}G).P_{XYZ}(F\times G \times H) \nonumber
\end{multline}
Then, if $P_{Y}(T_B^{-1}G)\neq 0$, $P_{XYZ}(T_A^{-1}F\times T_B^{-1}G \times T_C^{-1}H) = P_{XYZ}(F\times G \times H)$.

 If $P_{Y}(T_B^{-1}G) = 0$, since $P_Y$ is stationary, $P_{Y}(G) = 0$. Then $P_{XYZ}(T_A^{-1}F\times T_B^{-1}G \times T_C^{-1}H) = P_{XYZ}(F\times G \times H)=0$.

Then $P_{XYZ}$ is stationary on rectangles thus stationary on the $\sigma$-field generated by the rectangles. This implies that $P_{XZ}$ is stationary. The channel $[A,P_{Z|X},C]$ is then quasi-stationary.
\end{proof}

\begin{proposition}\label{PropositionCascadeOfRecurrentChannels}
Let $\nu = [A,P_{Y|X},B]$ and $\nu' = [B,P_{Z|Y},C]$ two channels in cascade. If $[B,P_{Z|Y},C]$ is a recurrent (or equivalently  incompressible) channel then the cascade $\nu\nu'=[A,P_{Z|X},C]$ is a recurrent channel.
\end{proposition}

\begin{proof}
Let $P_{X}$ be the distribution of a recurrent source.

$\forall H \in \mathcal{B}_{C^\mathcal{I}}$, such that $T_C^{-1}H \subset H$, assume that 
$$
P_{Z|Y}(T_C^{-1}H|y) = P_{Z|Y}(H|y)\text{  }P_Y\text{-a.e.}
$$
 Then $P_{Z|Y}(T_C^{-1}H|y) = P_{Z|Y}(H|y)$ $P_{XY}$-a.e. 

This implies that, for any $F\in \mathcal{B}_{A^\mathcal{I}}$:
\begin{multline}
\int_F P_{Z|X}(H|x) dP_X= \int_F \int  P_{Z|Y}(H|y) dP_{Y|X} dP_X  \\
			= \int_{F\times B^\mathcal{I}} P_{Z|Y}(H|y) dP_{XY} \\
			= \int_{F\times B^\mathcal{I}} P_{Z|Y}(T_C^{-1}H|y) dP_{XY} \\
			=  \int_F P_{Z|X}(T_C^{-1}H|x)  dP_X \nonumber
\end{multline}
Then
$$
P_{Z|X}(H|x) = P_{Z|X}(T_C^{-1}H|x) \text{  }P_X\text{-a.e.}
$$
By Proposition \ref{PropositionRecurrentChannel2}, the channel $[A,P_{Z|X},C]$ is recurrent w.r.t. $P_X$.
\end{proof}

\begin{proposition}\label{PropositionCascadeOfRAMSChannels}
Let $\nu = [A,P_{Y|X},B]$ and $\nu' = [B,P_{Z|Y},C]$ two R-AMS channels in cascade. Then 
\begin{enumerate}
	\item the cascade $\nu\nu'=[A,P_{Z|X},C]$ is an  R-AMS channel
	\item for any stationary source $[A,X]$ with distribution $\mu$, 
	$$
	\mu\nu\nu' \ll \mu\overline{\nu}_\mu\overline{\nu'}_{\overline{\eta}}
	$$
	 where $\overline{\eta}$ is the stationary output of the hookup $\mu \overline{\nu}_\mu$
	\item for any R-AMS source  $[A,X]$ with distribution $\mu$, 
	$$
	\mu\nu\nu' \ll \overline{\mu}\nu\nu' \ll \overline{\mu}\ \overline{\nu}_{\overline{\mu}}\ \overline{\nu'}_{\overline{\eta}}
	$$
	 where $\overline{\eta}$ is the stationary output of the hookup $\overline{\mu}\overline{\nu}_{\overline{\mu}}$
\end{enumerate}
\end{proposition}

\begin{proof}
Let $\mu$ be the distribution of a stationary source on alphabet $A$. $\nu$ is R-AMS, thus, by Proposition \ref{PropositionNecessarySufficientConditionRAMSChannel}, there exists a quasi-stationary channel $\overline{\nu}_\mu$ such that 
$$
\nu(x,.) \ll\overline{\nu}_{\mu}(x,.) \text{  }\mu\text{.a.e.}
$$
and the hookup $\mu\overline{\nu}_\mu$ is stationary which implies that its output marginal $\overline{\eta}$ is stationary.

$\overline{\eta}$ is stationary and $\nu'$ is R-AMS. By Proposition \ref{PropositionNecessarySufficientConditionRAMSChannel}, there exists a quasi-stationary channel $\overline{\nu'}_{\overline{\eta}}$ such that 
$$
\nu'(y,.) \ll\overline{\nu'}_{\overline{\eta}}(y,.) \text{  }\overline{\eta}\text{.a.e.}
$$
 and, by Proposition \ref{PropositionCascadeOfQuasiStationaryChannels}, the cascade $\overline{\nu}_\mu\overline{\nu'}_{\overline{\eta}}$ is quasi-stationary. Thus $\mu\overline{\nu}_\mu\overline{\nu'}_{\overline{\eta}}$ is stationary.

Let $O \in \mathcal{B}_{A^\mathcal{I}\times B^\mathcal{I}\times C^\mathcal{I}}$ such that 
$$
\mu\overline{\nu}_\mu\overline{\nu'}_{\overline{\eta}}(O)=0
$$
Then 
\begin{multline}
	 \int \overline{\nu'}_{\overline{\eta}}(y,O_{xy}) d\mu\overline{\nu}(x,y) = 0 \\
	\Rightarrow \overline{\nu'}_{\overline{\eta}}(y,O_{xy}) = 0  \text{   }\mu\overline{\nu}\text{.a.e.}\\
	\Rightarrow \nu'(y,O_{xy}) = 0  \text{   }\mu\overline{\nu}\text{.a.e.}\nonumber
\end{multline}
$\nu(x,.) \ll \overline{\nu}(x,.)$ $\mu$.a.e. implies $\mu\nu \ll \mu\overline{\nu}$. Then
$$
\nu'(y,O_{xy}) = 0  \text{   }\mu\nu\text{.a.e.}
$$
In other words $\mu\nu\nu'\ll \mu\overline{\nu}\overline{\nu'}_{\overline{\eta}}$. Then $\nu\nu'$ is R-AMS w.r.t. $\mu$. 

This closes the proof of the two first statements.

Let $\mu$ be an R-AMS source. $\mu$ is AMS and recurrent then $\mu\ll\overline{\mu}$. By Lemma \ref{Lemma2Fontana}, this implies that $\mu\alpha \ll \overline{\mu}\alpha$ for any channel $[A,\alpha,C]$.   If $\alpha =\nu\nu' $, from the two first statements of the proposition, it follows that
$$
\mu\nu\nu' \ll \overline{\mu}\nu\nu' \ll \overline{\mu}\ \overline{\nu}_{\overline{\mu}}\ \overline{\nu'}_{\overline{\eta}}
$$
\end{proof}

\begin{proposition}\label{PropositionCascadeOfAMSChannels}
Let $\nu = [A,P_{Y|X},B]$ and $\nu' = [B,P_{Z|Y},C]$ two AMS channels in cascade. 
\begin{enumerate}
	\item the cascade $\nu\nu'= [A,P_{Z|X},C]$ is an  AMS channel
	\item for any stationary source $[A,X]$ with distribution $\mu$, $\mu\nu\nu' \ll^a \mu\overline{\nu}_\mu\overline{\nu'}_{\overline{\eta}}$ where $\overline{\eta}$ is the stationary output of the hookup $\overline{\nu}_\mu$
	\item for any AMS source  $[A,X]$ with distribution $\mu$, $\mu\nu\nu' \ll^a \overline{\mu}\nu\nu' \ll^a \overline{\mu}\ \overline{\nu}_{\overline{\mu}}\ \overline{\nu'}_{\overline{\eta}}$ where $\overline{\eta}$ is the stationary output of the hookup $\overline{\mu}\overline{\nu}_{\overline{\mu}}$.
\end{enumerate}
\end{proposition}
\begin{proof}
The cascade $\nu_\infty \nu'_\infty$ is the restriction of the cascade $\nu\nu'$:  $\nu_\infty \nu'_\infty=(\nu\nu')_\infty$.  From Proposition \ref{PropositionCascadeOfRAMSChannels}, $(\nu\nu')_\infty$ is R-AMS thus $\nu\nu'$ is AMS. The other two statements derive from Lemma \ref{LemmaDominanceOnTailSigmaField} and Proposition  \ref{PropositionCascadeOfRAMSChannels}.
\end{proof}

\section{Quasi-stationary mean of an R-AMS channel with respect to a stationary source}\label{SectionQuasiStationaryMean}

In this section,  the quasi-stationary mean of an R-AMS channel w.r.t a {\em stationary} source is expressed as the limit of the Cesaro mean of a family of channels induced by the source, the channel and the shifts. First, it is proved that if, for a (non-recurrent) channel, this expression holds then the channel is AMS w.r.t the stationary source. Secondly, it is proved that this expression holds for any channel R-AMS w.r.t. a stationary source. These results and their proofs  generalize (and are adaptation of) those given by \cite{FontanaGrayKieffer81}.

Let $[A,X]$ be a stationary source with distribution $\mu$ and $[A,\nu,B]$ a channel. For any integer $i$, since $\mu$ is stationary, $\mu$ is the input marginal of $\mu\nu T_{AB}^{-i}$. Then there exists a channel $[A,\nu_i,B]$ such that 
$$
\mu\nu T_{AB}^{-i}= \mu\nu_i
$$

\begin{proposition}\label{PropositionAMSQuasiStationaryMean}
Let $[A,\nu,B]$ be a channel and $[A,X]$ be a stationary source with distribution $\mu$. If, for any  $G \in \mathcal{B}_{B^\mathcal{I}}$,  the limit 
$$
\overline{\nu}_{\mu}(x,G)= \lim_{n \to \infty} \frac{1}{n} \sum_{i=0}^{n-1} \nu_{i}(x,G)
$$
 exists $\mu$-a.e. then
\begin{itemize}
	\item the channel $[A,\nu,B]$ is AMS
	\item the channel $[A,\overline{\nu}_{\mu},B]$ is the quasi-stationary mean  of $[A,\nu,B]$ with respect to the stationary probability $\mu$.
\end{itemize}
\end{proposition}
\begin{proof}
Let $\mu$ be the distribution of a stationary source.

Assume that, for any $ G \in \mathcal{B}_{B^\mathcal{I}}$, the following limit exists $\mu$-a.e.
$$
\overline{\nu}_{\mu}(x,G)= \lim_{n \to \infty} \frac{1}{n} \sum_{i=0}^{n-1} \nu_i(x,G)
$$
By  the Vitali-Hahn-Saks theorem, $\overline{\nu}_{\mu}(x,.)$ is a probability on the measurable space $(B^\mathcal{I},\mathcal{B}_{B^\mathcal{I}})$. Moreover, for any $G\in B^\mathcal{I}$, the function $x\mapsto \overline{\nu}_{\mu}(x,G)$ is measurable. Then $[A,\overline{\nu}_\mu,B]$ is a channel.

$$
\forall F  \in \mathcal{B}_{A^\mathcal{I}}, \forall G \in \mathcal{B}_{B^\mathcal{I}}, \int_F \overline{\nu}_{\mu}(x,G) d\mu = \int_F \lim_{n \to \infty} \frac{1}{n} \sum_{i=0}^{n-1} \nu_i(x,G)
$$
For any $G$, $(x \mapsto \frac{1}{n} \sum_{i=0}^{n-1} \nu_i(x,G))_n$ is a sequence of bounded measurable functions of $x$ which converges, then:
\begin{multline}
\int_F \overline{\nu}_\mu(x,G) d\mu  =   \lim_{n \to \infty} \frac{1}{n} \sum_{i=0}^{n-1} \int_F \nu_i(x,G)   \nonumber \\
	=  \lim_{n \to \infty} \frac{1}{n} \sum_{i=0}^{n-1} \mu\nu_{i}(F\times G)   =  \lim_{n \to \infty} \frac{1}{n} \sum_{i=0}^{n-1} \mu\nu T_{AB}^{-i}(F\times G) \nonumber
\end{multline}
Then for any element $R$ of the field generated by rectangles $F\times G$, the limit $\overline{\mu\nu}(R)=\lim_{n \to \infty} \frac{1}{n} \sum_{i=0}^{n-1} \mu\nu T_{AB}^{-i}(R)$ exists. From  Caratheodory theorem, $\overline{\mu\nu}$ uniquely extends to a probability $\overline{\mu\nu}$ on the $\sigma$-field $\mathcal{B}_{A^\mathcal{I}\times A^\mathcal{I}}$ . Hence $\mu\nu$ is AMS. Moreover $\overline{\mu\nu}$ is stationary and, since $\mu$ is stationary,  $\overline{\mu\nu}=\mu\overline{\nu}_\mu$. This implies that $\overline{\nu}_\mu$ is quasi-stationary w.r.t $\mu$.
\end{proof}

\begin{proposition}\label{PropositionRAMSQuasiStationaryMean}
Let $[A,\nu,B]$ be a channel R-AMS w.r.t  a stationary source $\mu$. Then the quasi-stationary mean $[A,\overline{\nu}_{\mu},B]$  (w.r.t $\mu$) is such that
$$
\forall G \in \mathcal{B}_{B^\mathcal{I}}, \overline{\nu}_{\mu}(x,G)= \lim_{n \to \infty} \frac{1}{n} \sum_{i=0}^{n-1} \nu_{i}(x,G) \text{  }\mu \text{-a.e.}
$$ 
\end{proposition}
\begin{proof}

Let $[A,\nu,B]$ be an  R-AMS channel and $\mu$ be the distribution of a stationary source.  $\nu$ is AMS, then, $\forall F \in \mathcal{B}_{A^\mathcal{I}}$ and $\forall G \in \mathcal{B}_{B^\mathcal{I}}$
$$
\mu\overline{\nu}_{\mu}(F\times G)=\overline{\mu\nu}(F\times G)=\lim_{n \to \infty} \frac{1}{n} \sum_{i=0}^{n-1} \mu\nu T_{AB}^{-i}(F \times G)
$$
Since $[A,\nu_i,B]$ is the channel such that $\mu\nu_i = \mu\nu T_{AB}^{-i}$, this implies:
$$
\forall F \in \mathcal{B}_{A^\mathcal{I}}, \forall G \in \mathcal{B}_{B^\mathcal{I}}, \int_F \overline{\nu}_{\mu}(x,G) d\mu = \lim_{n \to \infty} \frac{1}{n} \sum_{i=0}^{n-1} \int_F \nu_i(x,G)  d\mu
$$

$\mu$ is stationary and $\nu$ R-AMS then $\mu\nu \ll \mu\overline{\nu}_{\mu}$ and $\forall i\geq 0$, $\mu\nu T_{AB}^{-i} \ll \mu\overline{\nu}_{\mu} T_{AB}^{-i}=\mu\overline{\nu}_{\mu}$. Then, $\forall i$, the Radon-Nikodym derivative $f_i=\frac{d\mu\nu_i}{d\mu\overline{\nu}_{\mu}}$ exists. 
$$
\forall F \in \mathcal{B}_{A^\mathcal{I}}, \forall G \in \mathcal{B}_{B^\mathcal{I}}, \int_F \nu_i(x,G) d\mu = \int_F \int_G f_i(x,y) d\overline{\nu}_{\mu} d\mu
$$
Thus 
\begin{equation}
\nu_i(x,G) = \int_G f_i(x,y) d\overline{\nu}_{\mu}\text{  }\mu\text{-a.e.} \label{EqNuI}
\end{equation}

Since $\forall i$, $f_i=\frac{d\mu\nu_i}{d\mu\overline{\nu}_{\mu}}$, from Theorem 7 of \cite{GrayKieffer80}
$$
\lim_{n \to \infty} \frac{1}{n} \sum_{i=0}^{n-1} f_i(x,y) = 1
$$
Then 
$$
\lim_{n \to \infty} 1_G(y).\frac{1}{n} \sum_{i=0}^{n-1} f_i(x,y) = 1_G(y)
$$
This implies (thanks to Fatou's lemma):
\begin{multline}
\overline{\nu}_{\mu}(x,G) = \int 1_G(y) d\overline{\nu}_{\mu}  \leq  \liminf_{n \to \infty} \int 1_G(y).\frac{1}{n} \sum_{i=0}^{n-1} f_i(x,y) d\overline{\nu}_{\mu} \nonumber \\
					 \leq   \liminf_{n \to \infty} \frac{1}{n} \sum_{i=0}^{n-1} \int_G f_i(x,y) d\overline{\nu}_{\mu}  \nonumber
\end{multline}
By (\ref{EqNuI})
$$
\liminf_{n \to \infty} \frac{1}{n} \sum_{i=0}^{n-1} \int_G f_i(x,y) d\overline{\nu}_{\mu} =  \liminf_{n \to \infty} \frac{1}{n} \sum_{i=0}^{n-1} \nu_i (x,G)
$$
Then
$$
\overline{\nu}_{\mu}(x,G) \leq \liminf_{n \to \infty} \frac{1}{n} \sum_{i=0}^{n-1} \nu_i (x,G)
$$

It also hold that:
\begin{multline}
\overline{\nu}_{\mu}(x,G^c) = 1 - \overline{\nu}_{\mu}(x,G) = \int 1_{G^c}(y) d\overline{\nu}_{\mu} \\
	\leq  \liminf_{n \to \infty} \int 1_{G^c}(y).\frac{1}{n} \sum_{i=0}^{n-1} f_i(x,y) d\overline{\nu}_{\mu}  \leq  \liminf_{n \to \infty} \frac{1}{n} \sum_{i=0}^{n-1} (1-\nu_i(x,G))  \\
		\leq  1 + \liminf_{n \to \infty} \frac{1}{n} \sum_{i=0}^{n-1} (-\nu_i(x,G))  \leq  1 - \limsup_{n \to \infty} \frac{1}{n} \sum_{i=0}^{n-1} \nu_i(x,G) \nonumber
\end{multline}
which implies 
$$
\overline{\nu}_{\mu}(x,G) \geq \limsup_{n \to \infty} \frac{1}{n} \sum_{i=0}^{n-1} \nu_i(x,G)
$$ 

Since  
$$
\overline{\nu}_{\mu}(x,G) \leq \liminf_{n \to \infty} \frac{1}{n} \sum_{i=0}^{n-1} \nu_i(x,G)
$$
then $\lim_{n \to \infty} \frac{1}{n} \sum_{i=0}^{n-1} \nu_i (x,G)$ exists and is equal to  $\overline{\nu}_{\mu}(x,G)$.
\end{proof}

\begin{lemma}\label{LemmaNuI}
$\forall G \in \mathcal{B}_{B^\mathcal{I}}$ and $\forall i$, $\nu_{i}(T_A^i x, G)=\nu(x,T_B^{-i} G)$ $\mu_{T_A^{-i}\mathcal{B}_{A^\mathcal{I}}}$-a.e.
\end{lemma}
\begin{proof}
For any $i$, the channel $\nu_i$ is such that ,  $\forall F  \in \mathcal{B}_{A^\mathcal{I}}$ and $\forall G \in \mathcal{B}_{B^\mathcal{I}}$:
\begin{eqnarray}
		\mu\nu_i( F \times G ) & = & \mu\nu T_{AB}^{-i}(F \times G ) \nonumber \\
\Rightarrow \int_F \nu_i(x,G) d\mu & = & \int_{T_A^{-i}F} \nu(x,T_B^{-i}G) d\mu \nonumber 
\end{eqnarray}
$\mu$ is stationary then
$$
\int_{T_A^{-i}F} \nu_i(T_A^ix,G) d\mu  =  \int_{T_A^{-i}F} \nu(x,T_B^{-i}G) d\mu
$$
Thus, if $\mu_{T_A^{-i}\mathcal{B}_{A^\mathcal{I}}}$ is the restriction of $\mu$ to the $\sigma$-field $T_A^{-i}\mathcal{B}_{A^\mathcal{I}}$, then:
$$
\nu_i(T_A^ix,G) = \nu(x,T_B^{-i} G)\text{   }\mu_{T_A^{-i}\mathcal{B}_{A^\mathcal{I}}} \text{ a.e.}
$$
\end{proof}
In the context of two-sided channels (i.e., considering invertible shifts), \cite{FontanaGrayKieffer81} proved that the stationary mean $\overline{\nu}$ of a two-sided AMS channel $\nu$ is given by:
$$
\overline{\nu}(x,G)=\lim_{n \to \infty} \frac{1}{n} \sum_{i=0}^{n-1} \nu(T_A^{-i}x,T_B^{-i}G)
$$
If the shift $T_A$ is invertible then $T_A^{-i}\mathcal{B}_{A^\mathcal{I}}=\mathcal{B}_{A^\mathcal{I}}$. Then, by Lemma \ref{LemmaNuI}:
$$
\forall i, \nu_i(T_A^ix,G)=\nu(x,T_B^{-i}G)\text{  }\mu\text{-a.e.}
$$
$T_A$ being invertible, substituting $T_A^{-i}x$ to $x$, it holds:
$$
\nu_i(x,G)=\nu(T_A^{-i}x,T_B^{-i}G)\text{  }\mu\text{-a.e.}
$$
which, by Proposition \ref{PropositionRAMSQuasiStationaryMean}, gives the result of  \cite{FontanaGrayKieffer81}.

\section{Ergodicity of one-sided AMS channels}\label{SectionErgodicAMSChannels}

This section is devoted to a proposition and its corollary which adapt to one-sided AMS channels some known ergodicity results on two-sided AMS channels: Theorem 6 of \cite{FontanaGrayKieffer81}, Theorem 3.7 of \cite{Kakihara91} and Theorem 4.4 of \cite{Kakihara03} (see also \cite{Kakihara99}). For the sake of completeness, it is worth citing Lemma 2.3 (page 28) of \cite{Gray11} (which gives a sufficient condition for an AMS channel to be ergodic)  and statement (6) of Theorem  4.4 of \cite{Kakihara03} which both hold for one-sided channels.

\begin{definition} \label{DefinitionErgodicChannel}
A quasi-stationary (resp. AMS) channel $[A,\nu,B]$ is ergodic with respect to a stationary (resp.  AMS) ergodic source $[A,X]$ with distribution $\mu$ if the hookup $\mu\nu$ is stationary (resp. AMS) and ergodic. A quasi-stationary (resp. AMS) channel $[A,\nu,B]$ is ergodic if it is ergodic  with respect to any stationary (resp. AMS) ergodic source.
\end{definition}

\begin{proposition}\label{PropositionErgodicityConditions}
Let  $[A,\nu,B]$ be an R-AMS channel. The following statements are equivalent:
\begin{enumerate}
	\item \label{PropositionErgodicityConditionsStatementErgodicity} $\nu$ is ergodic with respect to  ergodic R-AMS sources
	\item \label{PropositionErgodicityConditionsStatementErgodicityofQSmeans} for any ergodic stationary source  $[A,X]$ with distribution $\mu$, the quasi-stationary mean $\overline{\nu}_\mu$ is ergodic w.r.t $\mu$
	\item \label{PropositionErgodicityConditionsStatementDominanceByAnErgodicQSChannel} for any ergodic stationary source  $[A,X]$ with distribution $\mu$, there is an ergodic quasi-stationary channel $[A,\nu_\mu,B]$ such that 
	$$
	\nu(x,.) \ll \nu_\mu(x,.) \text{     }\mu\text{-a.e.}
	$$
	\item \label{PropositionErgodicityConditionsStatementDominanceByAnErgodicRAMSChannel} for any ergodic stationary source  $[A,X]$ with distribution $\mu$, there is an ergodic R-AMS channel $[A,\nu_\mu,B]$ such that 
	$$
	\nu(x,.) \ll \nu_\mu(x,.) \text{     }\mu\text{-a.e.}
	$$
\end{enumerate}
\end{proposition}
\begin{proof}
\ \\
\begin{itemize}
		\item[\ref{PropositionErgodicityConditionsStatementErgodicity} $\Rightarrow$ \ref{PropositionErgodicityConditionsStatementErgodicityofQSmeans}] 
			Let $\mu$ be the distribution of an ergodic stationary source. Then $\mu\nu$ is ergodic recurrent and AMS. $\mu\nu$ and $\overline{\mu\nu}=\mu\overline{\nu}_\mu$ coincide on invariant events thus $\mu\overline{\nu}_\mu$ is ergodic. Hence $\overline{\nu}_\mu$ is ergodic w.r.t $\mu$.
		\item[\ref{PropositionErgodicityConditionsStatementErgodicityofQSmeans} $\Rightarrow$ \ref{PropositionErgodicityConditionsStatementErgodicity}] 
			Let $\mu$ be  the distribution of an ergodic recurrent AMS source. Then $\mu\nu\ll \overline{\mu} \  \overline{\nu}_{\overline{\mu}}$.  By assumption  $\overline{\mu}\ \overline{\nu}_{\overline{\mu}}$  is ergodic, thus so is $\mu\nu$.
		\item[\ref{PropositionErgodicityConditionsStatementDominanceByAnErgodicQSChannel} $\Rightarrow$ \ref{PropositionErgodicityConditionsStatementErgodicityofQSmeans}] $\nu(x,.) \ll \nu_\mu(x,.) \text{     }\mu\text{-a.e.} \Rightarrow \mu\nu \ll \mu \nu_\mu$ and $\mu\nu_\mu$ is ergodic for any $\mu$ ergodic thus $\mu\nu$ is ergodic then $\mu\overline{\nu}_\mu$ is ergodic.
		\item[\ref{PropositionErgodicityConditionsStatementErgodicityofQSmeans} $\Rightarrow$ \ref{PropositionErgodicityConditionsStatementDominanceByAnErgodicQSChannel}] choose $\nu_\mu = \overline{\nu}_\mu$
		\item[\ref{PropositionErgodicityConditionsStatementDominanceByAnErgodicRAMSChannel} $\Rightarrow$ \ref{PropositionErgodicityConditionsStatementDominanceByAnErgodicQSChannel}] obvious because a quasi-stationary channel is R-AMS.
		\item[\ref{PropositionErgodicityConditionsStatementDominanceByAnErgodicQSChannel} $\Rightarrow$ \ref{PropositionErgodicityConditionsStatementDominanceByAnErgodicRAMSChannel}] $\nu_\mu$ is R-AMS then $\nu_\mu(x,.) \ll \overline{ \nu}_{\mu}(x,.)$ $\mu$.a.e. ; by transitivity of dominance, $\nu(x,.) \ll \overline{ \nu_{\mu}}(x,.)$ $\mu$.a.e.
\end{itemize}
\end{proof}

Since the invariant events belong to the tail $\sigma$-field, a probability $\mu$ is ergodic if and only if $\mu_\infty$ is ergodic. Thus $\mu$ is AMS and ergodic if and only if $\mu_\infty$ is R-AMS and ergodic. The following corollary follows immediately.
\begin{corollary}\label{CorollaryErgodicityConditions}
Let  $[A,\nu,B]$ be an AMS channel. The following statements are equivalent:
\begin{enumerate}
	\item \label{CorollaryErgodicityConditionsStatementErgodicity} $\nu$ is ergodic with respect to  ergodic AMS sources
	\item \label{CorollaryErgodicityConditionsStatementErgodicityofQSmeans} for any ergodic stationary source  $[A,X]$ with distribution $\mu$, the quasi-stationary mean $\overline{\nu}_\mu$ is ergodic w.r.t $\mu$
	\item \label{CorollaryErgodicityConditionsStatementDominanceByAnErgodicQSChannel} for any ergodic stationary source   $[A,X]$ with distribution $\mu$, there is an ergodic quasi-stationary channel $[A,\nu_\mu,B]$ such that 
	$$
	\nu(x,.) \ll^a \nu_\mu(x,.) \text{     }\mu_\infty \text{-a.e.}
	$$
	\item \label{CorollaryErgodicityConditionsStatementDominanceByAnErgodicRAMSChannel} for any ergodic stationary source   $[A,X]$ with distribution $\mu$, there is an ergodic AMS channel $[A,\nu_\mu,B]$ such that 
	$$
	\nu(x,.) \ll^a \nu_\mu(x,.) \text{     }\mu_\infty\text{-a.e.}
	$$
\end{enumerate}
\end{corollary}

\section{Quasi-stationary mean of an ergodic R-AMS channel with respect to an ergodic R-AMS source}\label{SectionQuasiStationaryMeanErgodicRAMSChannel}

In this section, it is proved that, considering an ergodic R-AMS channel $[A,\nu,B]$, its quasi-stationary mean   w.r.t an ergodic R-AMS source with distribution $\mu$ is equal to its quasi-stationary mean w.r.t $\overline{\mu}$: $\overline{\nu}_\mu=\overline{\nu}_{\overline{\mu}}$. A consequence is that Proposition \ref{PropositionRAMSQuasiStationaryMean} gives an expression of $\overline{\nu}_\mu$. Moreover, if $\mu_1$ and $\mu_2$ are two ergodic R-AMS source distributions, then  $\overline{\nu}_{\mu_1}$ and $\overline{\nu}_{\mu_2}$ are either identical or mutually singular.

\begin{proposition}\label{PropositionQSMofErgodicRAMSChannel}
Let  $[A,\nu,B]$ be an ergodic R-AMS channel. Then, the quasi-stationary mean of $\nu$ with respect to an ergodic and R-AMS source $[A,X]$ with distribution $\mu$ equals $\overline{\mu}$-a.e. the quasi-stationary mean of $\nu$ with respect to the stationary mean $\overline{\mu}$ of $\mu$:
$$
\overline{\nu_{\overline{\mu}}}=\overline{\nu_\mu} \text{ ~~~~~}\overline{\mu}\text{-a.e.  and }\mu\text{-a.e.}
$$
\end{proposition}
\begin{proof}
Let $\mu$ be the distribution of an ergodic R-AMS source and $[A,\nu,B]$ an ergodic R-AMS channel. $\mu$ is dominated by its stationary mean: $\mu \ll \overline{\mu}$. Then, by Lemma~\ref{Lemma2Fontana}, $\mu\nu \ll \overline{\mu}\nu$. $\nu$ being R-AMS and $\overline{\mu}$ being stationary, $\overline{\mu}\nu \ll \overline{\overline{\mu}\nu }= \overline{\mu}\ \overline{\nu_{\overline{\mu}}}$. Thus $\mu\nu \ll \overline{\mu}\ \overline{\nu_{\overline{\mu}}}$

Moreover $\nu$ being R-AMS and $\mu$ being R-AMS, $\mu\nu \ll \overline{\mu\nu}= \overline{\mu}\ \overline{\nu_\mu}$.

Then $\mu\nu$ is dominated by $\overline{\mu}\ \overline{\nu_{\overline{\mu}}}$ and $\overline{\mu}\ \overline{\nu_\mu}$ which are ergodic and stationary probabilities on the space $(A^\mathcal{I}\times B^\mathcal{I}, \mathcal{B}_{A^\mathcal{I}\times B^\mathcal{I}})$. By \cite{Kakihara99}, Lemma 1, page 75, two ergodic and stationary probabilities on the same space are either identical or mutually singular. The two probabilities $\overline{\mu}\ \overline{\nu_{\overline{\mu}}}$ and $\overline{\mu}\ \overline{\nu_\mu}$ dominate the same probability so they cannot be mutually singular. They are thus identical: $\overline{\mu}\ \overline{\nu_{\overline{\mu}}}=\overline{\mu}\ \overline{\nu_\mu}$. This implies:
$$
\overline{\nu_{\overline{\mu}}}=\overline{\nu_\mu} \text{  }\overline{\mu}\text{-a.e.  and }\mu\text{-a.e.}
$$
\end{proof}

\begin{proposition}
Let  $[A,\nu,B]$ be an ergodic R-AMS channel. Let $[A,X_1]$ and $[A,X_2]$ be two ergodic R-AMS sources with respective distributions $\mu_1$ and $\mu_2$. Then
\begin{itemize}
	\item  either the quasi-stationary means of $\nu$ with respect to $\mu_1$ and $\mu_2$ are equal and this holds if and only if the stationary means of $\mu_1$ and $\mu_2$ are equal
	$$
	\overline{\mu_1}=\overline{\mu_2} \Leftrightarrow \overline{\nu_{\mu_1}} = \overline{\nu_{\mu_2}}
	$$
	\item or the quasi-stationary means of $\nu$ w.r.t $\mu_1$ and $\mu_2$ are mutually singular $\overline{\mu_1}$-a.e and $\overline{\mu_2}$-a.e and this holds if and only if  the stationary means of $\mu_1$ and $\mu_2$ are mutually singular
	$$
	\overline{\mu_1} \perp \overline{\mu_2} \Leftrightarrow \overline{\nu_{\mu_{1}}}(x,.) \perp \overline{\nu_{\mu_{2}}}(x,.)\text{  } \overline{\mu_1}\text{-a.e. and }\overline{\mu_2}\text{-a.e.}
	$$
\end{itemize}
\end{proposition}
\begin{proof}
Let $\mu_1$ and $\mu_2$ be the distributions of two ergodic R-AMS sources. Their ergodic stationary means $\overline{\mu_1}$  and $\overline{\mu_2}$ are either identical or mutually singular.

 If $\overline{\mu_1} = \overline{\mu_2}$, thanks to Proposition \ref{PropositionQSMofErgodicRAMSChannel}, $\overline{\nu_{\mu_1}}=\overline{\nu_{\overline{\mu_1}}}=\overline{\nu_{\overline{\mu_2}}}=\overline{\nu_{\mu_2}}$. Moreover, being ergodic and stationary probabilities on the same space, $\overline{\mu_1}\ \overline{\nu_{\mu_1}}$ and $\overline{\mu_2}\ \overline{\nu_{\mu_2}}$ are either identical or mutually singular.

It should be noticed that $\overline{\mu_1} = \overline{\mu_2}$ if and only if $\overline{\mu_1}\ \overline{\nu_{\mu_1}} = \overline{\mu_2}\ \overline{\nu_{\mu_2}}$ and thus $\overline{\mu_1} \perp \overline{\mu_2}$ if and only if $\overline{\mu_1}\ \overline{\nu_{\mu_1}} \perp \overline{\mu_2}\ \overline{\nu_{\mu_2}}$

Assume that $\overline{\mu_1} \perp \overline{\mu_2}$. Then $\overline{\mu_1}\ \overline{\nu_{\mu_1}} \perp \overline{\mu_2}\ \overline{\nu_{\mu_2}}$. There exists a set $\Omega \in \mathcal{B}_{A^\mathcal{I}\times B^\mathcal{I}}$ such that $\overline{\mu_1}\ \overline{\nu_{\mu_1}}(\Omega)=1$ and $\overline{\mu_2}\ \overline{\nu_{\mu_2}}(\Omega)=0$. Thus
$$
\overline{\nu_{\mu_{1}}}(x,\Omega_x)=1 \text{  }\overline{\mu_1}\text{-a.e.}
$$
and
$$
\int \int \overline{\nu_{\mu_{2}}}(x,\Omega_x) d\overline{\mu_1}d\overline{\mu_2}=\int \int \overline{\nu_{\mu_{2}}}(x,\Omega_x) d\overline{\mu_2}d\overline{\mu_1} = \int \overline{\mu_2}\ \overline{\nu_{\mu_2}}(\Omega) d\overline{\mu_1} = 0 \nonumber
$$
This implies that 
$$
\int \overline{\nu_{\mu_{2}}}(x,\Omega_x) d\overline{\mu_1}=0
$$
 and then 
$$
\overline{\nu_{\mu_{2}}}(x,\Omega_x)=0 \text{  }\overline{\mu_1}\text{-a.e.}
$$
Symmetrically, 
$$
\overline{\nu_{\mu_{1}}}(x,\Omega_x^c)=0 \text{  }\overline{\mu_2}\text{-a.e.}
$$
 and 
$$
\overline{\nu_{\mu_{2}}}(x,\Omega_x^c)=1 \text{  }\overline{\mu_2}\text{-a.e.}
$$
 which gives $\overline{\nu_{\mu_{1}}}(x,\Omega_x)=1 \text{  }\text{  and  }\overline{\nu_{\mu_{2}}}(x,\Omega_x)=0 \text{ }\overline{\mu_1}\text{-a.e.} \text{  and  }\overline{\mu_2}\text{-a.e.}$ or equivalently $\overline{\nu_{\mu_{1}}}(x,.) \perp \overline{\nu_{\mu_{2}}}(x,.)$ $\overline{\mu_1}$ and $\overline{\mu_2}$-a.e. 
\end{proof}

A remark is that, in case of non-singular shift,  the quasi-stationary mean of a non-ergodic R-AMS channel $\nu$ w.r.t an R-AMS probability $\mu$ is equivalent to the quasi-stationary mean of $\nu$ w.r.t to the stationary mean $\overline{\mu}$ of $\mu$:
$\mu\nu \equiv \overline{\mu}\ \overline{\nu}_\mu$ 
and (since $\overline{\mu}\equiv \mu$) $\overline{\mu}\nu \equiv \mu\nu$ 
and $\overline{\mu}\nu \equiv \overline{\mu}\ \overline{\nu}_{\overline{\mu}}$,
 then $\overline{\mu}\ \overline{\nu}_\mu \equiv \overline{\mu}\ \overline{\nu}_{\overline{\mu}}$ 
which implies $\overline{\nu}_{\mu_x}\equiv \overline{\nu}_{\overline{\mu}_x} $ $\overline{\mu}$.a.e. and $\mu$.a.e.

\appendix
\section{Proof of Lemma \ref{LemmaRecurrenceOnRectangles}}\label{SectionProofsOfLemmas}

Proof of Lemma \ref{LemmaRecurrenceOnRectangles} relies on the following two lemmas.
\begin{lemma}\label{LemmaCountableUnionOfRecurrentSets}
Let $\mathcal{B}'$ be the set $\{ O \in \mathcal{B} / \mu(O \setminus \cup_{k\geq 1} T^{-k}O) = 0 \}$ where $\mathcal{B}$ is a $\sigma$-field of a sequence space. Then $\mathcal{B}'$ is stable by countable union.
\end{lemma}
\begin{proof}
Let $(O_i)_i$ be a countable family of elements of $\mathcal{B'}$. Then 
$$
\mu( (\cup_{i\geq 0} O_i) \setminus (\cup_{k\geq 1} T^{-k} (\cup_{i\geq 0} O_i))  = \mu( (\cup_{i\geq 0} O_i) \setminus (\cup_{i\geq 0} \cup_{k\geq 1}   T^{-k} O_i))
$$
$(A \cup B) \setminus (C\cup D) \subset (A \setminus C) \cup (B \setminus D)$, then:
$$
 \mu( (\cup_{i\geq 0} O_i) \setminus (\cup_{k\geq 1} T^{-k} (\cup_{i\geq 0} O_i)) \leq \mu( \cup_{i\geq 0} (O_i \setminus \cup_{k\geq 1}   T^{-k} O_i) )
$$
$$
\mu( (\cup_{i\geq 0} O_i) \setminus (\cup_{k\geq 1} T^{-k} (\cup_{i\geq 0} O_i))   \leq \sum_{i\geq 0} \mu( O_i \setminus \cup_{k\geq 1}   T^{-k} O_i) = 0
$$
Thus $\cup_{i\geq 0} O_i \in \mathcal{B}'$.
\end{proof}

\begin{lemma}\label{LemmaCountableUnionOfIncompressibleEvents}
Let $(\Omega, \mathcal{B}, T, \eta)$ a dynamical system. Let $\mathcal{B}'$ be the set $\{ E \in \mathcal{B}/ T^{-1}E \subset E \text{ and } \eta(T^{-1}E)=\eta(E) \}$. Then $\mathcal{B'}$ is stable by countable union.  
\end{lemma}
\begin{proof}
 Let $(E_i)_i$ be a countable family of elements of $\mathcal{B'}$. Then
$$
T^{-1}(\cup_i E_i) = \cup_i T^{-1}(E_i) \subset \cup_i E_i
$$
and 
$$
\eta((\cup_i E_i) \setminus \cup_i T^{-1}(E_i) ) \leq \eta(\cup_i (E_i \setminus T^{-1}(E_i) ) \leq \sum_i \eta(E_i \setminus T^{-1}(E_i))=0
$$
Then $\cup_i E_i \in \mathcal{B'}$
\end{proof}

\begin{proof}[Proof of Lemma \ref{LemmaRecurrenceOnRectangles}]
~\\
\begin{description}
	\item[(\ref{ItemRecurrenceOnGeneratingSets}) $\Rightarrow$ (\ref{ItemRecurrenceOnCountableUnionsOfFieldSets})] Any element of the field $\mathcal{F}$ is a finite union of rectangles. Then any countable union of field elements is a countable union of rectangles. By (\ref{ItemRecurrenceOnGeneratingSets}) and Lemma \ref{LemmaCountableUnionOfRecurrentSets},  (\ref{ItemRecurrenceOnCountableUnionsOfFieldSets}) holds. 

	\item[(\ref{ItemRecurrenceOnCountableUnionsOfFieldSets}) $\Rightarrow$ (\ref{ItemRecurrence})] 
Let $O\in \mathcal{B_{A^{\mathcal{I}}\times B^{\mathcal{I}}}}$. 
 The probability $\eta$ on  $( A^{\mathcal{I}}\times B^{\mathcal{I}} , \mathcal{B_{A^{\mathcal{I}}\times B^{\mathcal{I}}}})$ is the extension of  the set function $\eta$ on the field generated by the rectangles and verifies:
$$
\eta(O)=\inf_{(R_i)_{i\geq 0}/ O\subset \cup_{i\geq 0}R_i}\eta(\cup_{i\geq 0}R_i)
$$
where the families  $(R_i)_{i\geq 0}$ are countable covers of $O$ made of elements of the field generated by the rectangles (see \cite{Gray09}). Let $\epsilon >0$, then there exist  countable families of field elements $(R_i)_{i\geq 0}$ and $(R'_i)_{i\geq 0}$ respectively covering $O$ and $O^c$ such that:
$$
\eta(\cup_{i\geq 0}R_i)-\frac{\epsilon}{2} < \eta(O) \leq \eta(\cup_{i\geq 0}R_i) 
$$
$$
\eta(\cup_{i\geq 0}R'_i)-\frac{\epsilon}{2} < \eta(O^c) \leq \eta(\cup_{i\geq 0}R'_i) 
$$Let $\alpha=\cup_{i\geq 0}R_i$ and $\beta=\cup_{i\geq 0}R'_i$.  Obviously $T_{AB}^{-k}\beta^c \subset T_{AB}^{-k}O \subset T_{AB}^{-k}\alpha$ for any $k$. Then:
$$
O\setminus \cup_{k\geq 1} T_{AB}^{-k}O \subset   \alpha \setminus \cup_{k\geq 1} T_{AB}^{-k}\beta^c \subset (\beta^c \setminus \cup_{k\geq 1} T_{AB}^{-k}\beta^c) \cup (\alpha \setminus \beta^c)
$$
$\beta^c$ is a countable union of  elements of the field generated by rectangles then, by (\ref{ItemRecurrenceOnCountableUnionsOfFieldSets}), $\beta^c$ is a recurrent event. Moreover, $\eta(\alpha \setminus \beta^c) < \epsilon$. Then $\eta(O\setminus \cup_{k\geq 1} T_{AB}^{-k}O)<\epsilon$. This holds for any $\epsilon>0$, then 
$$
\eta(O\setminus \cup_{k\geq 1} T_{AB}^{-k}O)=0
$$
The event $O$ is recurrent.

	\item[(\ref{ItemRecurrence}) $\Leftrightarrow$ (\ref{ItemIncompressibility})] This is Theorem 7.3 of \cite{Gray09} (page 218).

	\item[(\ref{ItemIncompressibility}) $\Rightarrow$ (\ref{ItemIncompressibilityOnGeneratingSets})] This is obvious.

	\item[(\ref{ItemIncompressibilityOnGeneratingSets}) $\Rightarrow$ (\ref{ItemIncompressibilityOnCountableUnionsOfFieldSets})] Any element of the field $\mathcal{F}$ is a finite union of rectangles. Then any countable union of field elements is a countable union of rectangles. By (\ref{ItemIncompressibilityOnGeneratingSets}) and Lemma \ref{LemmaCountableUnionOfIncompressibleEvents},  (\ref{ItemIncompressibilityOnCountableUnionsOfFieldSets}) holds. 

	\item[(\ref{ItemIncompressibilityOnCountableUnionsOfFieldSets}) $\Rightarrow$ (\ref{ItemRecurrenceOnCountableUnionsOfFieldSets})]

Let $\mathcal{B''}$ be the set of countable unions of elements of the field $\mathcal{F}$. Let $E \in \mathcal{B''}$. Let $E^*= \cup_{i\geq 1}T_{AB}^{-1}E$. For any rectangle $F\times G$, $T_{AB}^{-1}F\times G = T_A^{-1} F \times T_B^{-1}G$ is also a rectangle. Then $E^* \in \mathcal{B''}$. Moreover,  $T_{AB}^{-1}( E^*) \subset E^*$. Let $E'=E \cup E^*$.  $E' \in \mathcal{B''}$ and $T_{AB}^{-1}( E') \subset E'$. Then, by (\ref{ItemIncompressibilityOnCountableUnionsOfFieldSets}):
$$
\eta(E' \setminus T_{AB}^{-1}( E'))=0
$$
Then
$$
\eta(E \cup E^*  \setminus E^*)=\eta(E \setminus E^*)=0
$$
Then (\ref{ItemRecurrenceOnCountableUnionsOfFieldSets}) holds.

	\item[(\ref{ItemRecurrenceOnCountableUnionsOfFieldSets}) $\Rightarrow$ (\ref{ItemRecurrenceOnGeneratingSets})] This is obvious.
\end{description}
\end{proof}


\section*{Acknowledgments}

The author wishes to thank  the referee and Joseph Kung  for very helpful remarks and suggestions that considerably improved the readability of this article and Jean-Paul Allouche for his encouragements.

\end{document}